\documentclass[1p,12pt]{elsarticle}
\usepackage{amsthm,amsmath,amssymb}
\usepackage{graphicx,color}
\usepackage{multicol}
\usepackage{import}
\usepackage{mathtools}
\usepackage{tikz}
\usepackage{enumitem}

\usetikzlibrary{decorations.pathreplacing,decorations.markings}
\usetikzlibrary{shapes,arrows}
\usepackage{footmisc}

\tikzset{
  on each segment/.style={
    decorate,
    decoration={
      show path construction,
      moveto code={},
      lineto code={
        \path [#1]
        (\tikzinputsegmentfirst) -- (\tikzinputsegmentlast);
      },
      curveto code={
        \path [#1] (\tikzinputsegmentfirst)
        .. controls
        (\tikzinputsegmentsupporta) and (\tikzinputsegmentsupportb)
        ..
        (\tikzinputsegmentlast);
      },
      closepath code={
        \path [#1]
        (\tikzinputsegmentfirst) -- (\tikzinputsegmentlast);
      },
    },
  },
  mid arrow/.style={postaction={decorate,decoration={
        markings,
        mark=at position .5 with {\arrow[#1]{stealth}}
      }}},
}

\newtheorem{theorem}{Theorem}[section]
\newtheorem{cor}[theorem]{Corollary}
\newtheorem{prop}[theorem]{Proposition}
\newtheorem{lemma}[theorem]{Lemma}

\newtheorem{problem}[theorem]{Riemann-Hilbert problem}
\newtheorem{pproblem}[theorem]{Mixed $\dbar$-Riemann-Hilbert problem}
\theoremstyle{definition}

\newtheorem{remark}[theorem]{Remark}

\renewcommand{\Re}{\operatorname{Re}}
\renewcommand{\Im}{\operatorname{Im}}
\newcommand{\dbar}{\bar{\partial}}

\newcommand{\R}{\mathbb{R}}
\newcommand{\Z}{\mathbb{Z}}
\newcommand{\C}{\mathbb{C}}
\newcommand{\K}{\mathcal{K}}
\newcommand{\OO}{\mathcal{O}}
\newcommand{\ad}{\text{ ad }}
\newcommand{\TU}[1]{\begin{pmatrix}
1&#1\\
0&1
\end{pmatrix}}
\newcommand{\TL}[1]{\begin{pmatrix}
1&0\\
#1&1
\end{pmatrix}}
\begin{document}
\begin{frontmatter}
\title
{A $\dbar$-steepest descent method for oscillatory Riemann-Hilbert problems}
\author[fd]{Fudong Wang\fnref{cor1}}
\author[ma1,ma2,ma3,ma4]{Wen-Xiu Ma\fnref{cor1}}
\address[fd]{Department of Mathematics and Statistics, \\ University of South Florida}
\address[ma1]{Department of Mathematics,\\ Zhejiang Normal University,
Jinhua 321004, Zhejiang, China}
\address[ma2]{Department of Mathematics, \\ King Abdulaziz University, Jeddah 21589, Saudi Arabia }
\address[ma3]{Department of Mathematics and Statistics, \\ University of South Florida}
\address[ma4]{School of Mathematical and Statistical Sciences,\\ North-West University, Mafikeng Camus, \\Private Bag X2046, Mmabatho 2735, South Africa}

\fntext[cor1]{Corresponding authors: mawx@cas.usf.edu (Wen-Xiu Ma), fudong@usf.edu (Fudong Wang).}

\begin{abstract}
  We study the long-time asymptotic behavior of oscillatory Riemann-Hilbert problems (RHPs) arising in the mKdV hierarchy (reducing from the AKNS hierarchy). Our analysis is based on the idea of $\dbar$-steepest descent. We consider RHPs generated from the inverse scattering transform of the AKNS hierarchy with weighted Sobolev initial data. The asymptotic formula for three regions of the spatial and temporal dependent variables  are presented in details.
\end{abstract}
\end{frontmatter}

\section{Introduction}
The long-time behavior of solutions of the initial-value problem for nonlinear evolution integrable PDEs has been studied extensively. It is well-known that the long-time asymptotic analysis for the integrable PDEs can be, via inverse scattering, formulated as a problem of finding asymptotics of certain oscillatory RHPs. A countless number of papers (see e.g. \cite{do2009nonlinear,varzugin96} and the references therein) have been devoted to studying the asymptotic behavior of a certain type of oscillatory 2 by 2 matrix RHPs, which is also the main subject of the current study. The most influential is the nonlinear steepest descent method (or the Deift-Zhou method), which was published in {\it{Annals of mathematics}} \cite{DZ93} in 1993. Before the Deift-Zhou work, A.R.Its \cite{its1981asymptotics} proposed a direct method, via an isomonodromic deformation, to study the asymptotics of a RHP arising in studying the long-time behavior of the nonlinear Schr\"odinger (NLS) equation. Ten years after the Deift-Zhou method was published, Deift and Zhou extended their method to study the long-time behavior of the defocusing NLS equation on some weight Sobolev space. Between 1993 and 2003, the Deift-Zhou method had been applied to not only the long-time behavior of integrable systems, but also equilibrium measure for logarithmic potentials \cite{deift1998new}, the strong asymptotics of orthogonal polynomials \cite{Deift1999} and many other other fields in mathematical physics. Shortly after the Deift-Zhou 2003 paper, McLaughlin and Miller \cite{mclaughlin2004dbar} proposed another generalization to the Deift-Zhou method: the so-called $\dbar$-steepest descent. This method was first applied to studying the long-time asymptotics of the defocusing NLS equation in 2008 \cite{dieng2008longtime}, see also its extension version \cite{dieng2018dispersive}. Comparing to the Deift-Zhou method, the $\dbar$-steepest descent method provides a more elementary way and more tractable way of analyzing the error terms. Since then, the $\dbar$-steepest descent method has been applied to many long-time asymptotic studies for nonlinear integrable PDEs, such as the focusing NLS equation \cite{borghese2016long}, the KdV equation \cite{giavedoni2017long}, the mKdV equation\cite{chen2019longtime}, the sine-Gordon equation\cite{chen2020longtime}, the fifth order mKdV equation \cite{liu2019longtime} and many others. It is worth mentioning that the mKdV and fifth-order mKdV equations belong to the mKdV hierarchy we will consider in the current work. In fact, by carefully checking \cite{chen2019longtime} and \cite{liu2019longtime}, we find there are many similar analyses which motivate us to study the whole mKdV hierarchy at once.

In the current paper, we will study an oscillatory 2 by 2 matrix RHP arising in studying the long-time asymptotics of the mKdV hierarchy. We will discuss the defocusing case (i.e., without solitons). The focusing case will be treated somewhere else in the future. The main analysis is based on the idea of $\dbar$-steepest descent \cite{dieng2008longtime,mclaughlin2004dbar}. 

In the study of Cauchy initial-value problems of integrable systems by means of inverse scattering, the following RHP appears: 
\begin{problem}
\label{the first RHP}
Looking for a 2 by 2 matrix-valued function $m(z)$ such that
\begin{enumerate}[label={(\arabic*)}]
    \item $m$ is analytic off the real line $\R$;
    \item for $z\in \R$, we have \begin{align}
m_+=m_-v_{\theta}(z),\quad  z\in \R,
\end{align}
where $m_{\pm}(z)=\lim_{\epsilon\rightarrow 0^{+}}m(z\pm i\epsilon),z\in\R$, and the jump reads
\begin{align}
v_{\theta}(z)=\begin{pmatrix}
 1-|R(z)|^2 & -\bar{R}(z)e^{-2it\theta}\\
 R(z)e^{2it\theta} & 1
 \end{pmatrix},
 \label{jump matrix}
\end{align}
where $R(z)$ is the reflection coefficient in performing inverse scattering with given initial data, see, e.g., \eqref{reflection coeff}, and $\theta=\theta(z;x/t)$ is a polynomial of $z$ with coefficients depends on $x/t$; 
\item $m(z)=I+\OO(z^{-1}),\quad z\rightarrow\infty.$
\end{enumerate}
\end{problem}

In this paper, we consider the following defocusing mKdV type reduction of the AKNS hierarchy (shortly, we call it the mKdV hierarchy): Fixing $n$ as an positive odd integer, we consider 
\begin{align}
\psi_x(x,t;z)&=\left(-iz\sigma_3+\begin{pmatrix}
      0&q(x,t)\\
      q(x,t)&0
    \end{pmatrix}\right)\psi(x,t;z),\\
    \psi_t(x,t;z)&=\left(\sum_{k=0}^nQ_k(x,t)z^k\right)\psi(x,t;z),
\end{align}
where $\sigma_3=\text{diag}(1,-1)$, $q(x,t)$ is the potential which solves a certain $1+1$ dimensional integrable equation, and $Q_k$ is determined by certain recursion relation (for details, see \cite{ablowitz1991}). 

In the case of the mKdV hierarchy, $Q_n(x,t)$ is a constant with respect to $x,t$. The corresponding nonlinear integrable PDE is worked out by the the zero curvature condition, which is also equivalent to $\psi_{xt}=\psi_{tx}$. In this paper, we will study the Cauchy initial-value problem for integrable PDEs generated from the defocusing mKdV hierarchy, with the initial data belonging to $H^{1,1}(\R)=\{f\in L^2|f'\in L^2, xf\in L^2\}$. Due to Zhou's result \cite{zhou98}, after direct scattering, the reflection coefficient $R(z)$ also belongs to $H^{1,1}(\R)$. By performing the time evolution, we arrive at the RHP \ref{the first RHP}. The first part (oscillating region) of the analysis is slightly more general than the one in the AKNS hierarchy, by making the following assumptions on the phase function: 
\begin{enumerate}[label={(\arabic*)}]
 \item $\theta$ is a real polynomial of degree $n$ with respect to $z$, with coefficients depends on $x/t$;
 \item $\theta'(z_j)=0,\theta''(z_j)\neq 0$ for $j=1,\cdots , l$, where $l$ denotes the number of real stationary phase points.
 \end{enumerate}
 \begin{remark}
     For the defocusing mKdV hierarchy case, $n$ in the first assumption corresponds to the $\frac{n-1}{2}$th member of the hierarchy. Since in mKdV hierarchy, $n$ is an odd number, say $n=2k-1,k\in\Z_+$, we will only need to study the phase function of the type: $c_1z+c_2z^{2k+1},\  k\in \Z_+,$ and $c_1,c_2$ are some constants. The purpose of the second assumption includes the case of linear combination of several members in the mKdV hierarchy, which is again integrable. In such situation, we will see a generic polynomial of $z$ with coefficients depends on $x/t$.
 \end{remark}

 \subsection{Main results}
 Before we establish the main results, we first introduce some notations. Let's denote the weighted Sobolev space by
 \begin{align*}
     H^{k,j}(\R):=\{f(x)\in L^2(\R):\partial_x^sf\in L^2(\R),\ s=1,\cdots,k,\ x^jf(x)\in L^2(\R)\},
 \end{align*}
 with norm
 \begin{align*}
     \|f\|_{H^{k,j}}:=\left(\|f\|^2_{L^2}+\sum_{l=1}^k\|\partial_x^lf\|^2_{L^2}+\sum_{m=1}^j\|x^mf\|^2_{L^2}\right)^{1/2}.
 \end{align*}
 Next we define the meaning of the long-time behavior in the three regions we are concerned with as follows.
 \begin{enumerate}[label={(\arabic*)}]
     \item Long-time behavior of the potential in the oscillating region means taking the $t\rightarrow\infty$ limit of $q(x,t)$ along the ray $x=-ct,c>0,t\rightarrow\infty$.
     \item Long-time behavior of the potential in the fast decaying region means taking the $t\rightarrow\infty$ limit of $q(x,t)$ along the ray $x=ct,c>0,t\rightarrow\infty$.
     \item Long-time behavior of the potential in the Painlev\'e region means taking the $t\rightarrow\infty$ limit of $q(x,t)$ along the curve $x=c(nt)^{1/n},c\neq 0,\, t\rightarrow\infty$, where $n$ is the degree of the polynomial phase function $\theta$.
 \end{enumerate}
\begin{theorem}
\label{Theorem 1}
In the oscillating region, provided that the initial data\footnote{Due to Zhou's theorem \cite{zhou98}, $R(z)$ belongs to $H^{1,n-1}(dz)$, then the time evolving reflection coefficient $R(z)e^{\pm 2it\theta}$ will stay in $H^{1,1}(dz)$ since the degree of $\theta$ is $n$. \label{fn:n} } $q(x,0)\in H^{n-1,1}(\R,dx)$, the long-time behavior for the potentials $q(x,t)$ reads
\begin{align}
q(x,t)&=q_{as}(x,t)+\OO(t^{-3/4}),t\rightarrow \infty,
\end{align}
where
\begin{align*}
q_{as}(x,t)&=-2i\sum_{j=1}^{l}\frac{|\eta(z_j)|^{1/2}}{\sqrt{2t\theta''(z_j)}}e^{i\varphi(t)},\\
\varphi(t)&=\frac{\pi}{4}-\arg\Gamma(-i\eta(z_j))\\
&-2t\theta(z_j)-\frac{\eta(z_j)}{2}\log|2t\theta''(z_j)|+2\arg(\delta_j)+\arg(R_j),
\end{align*}
and the phase function $\theta$ will depend only on  $z$ along any ray in the oscillating region,
$\{z_j\}_{j=1}^l$ are the real stationary phase points of the phase function, and 
\begin{align*}
    \delta(z)&=\exp{\left(\frac{1}{2\pi i}\int_{D_-}\frac{\log(1-|R(s)|^2)}{s-z}\text{ds}\right)},z\in \C\backslash D_-, \\
    D_-&=\{z\in\R:\theta'(z)<0\},\\
    \eta(z)&=-\frac{1}{2\pi}\log(1-|R(z)|^2),\\
    R_j&=R(z_j),\quad j=1,..,l,\\
    \delta_j&=\lim_{\substack{z=z_j+\rho e^{i\phi},\\ \rho\rightarrow 0,\\ \text{fix } \phi\in (0,\pi/2)}}\delta(z)(z-z_j)^{i\eta(z_j)}.
\end{align*}
Here $R(z)$ is the reflection coefficient generated from the standard inverse scattering procedure, see equation \eqref{reflection coeff}.
\end{theorem}
\begin{cor}
For the case of the AKNS hierarchy, in the oscillating region, the phase function $\theta(z)=\frac{x}{t}z+cz^{n},c>0$ and has just two real stationary phase points: $z_{\pm}=\pm\left|-\frac{x}{nct}\right|^{\frac{1}{n-1}}$, and 
then the long-time asymptotics for the potentials in the AKNS hierarchy are merely a special case of  Theorem \ref{Theorem 1}.
\end{cor}

\begin{theorem}
\label{theorem fast decay}
In the fast decay region, the long-time behavior for the potential reads
\begin{align}
q(x,t)=\OO(t^{-1}), \quad t\rightarrow \infty.
\end{align}
\end{theorem}
\begin{theorem}
\label{theorem painleve}
In the Painlev\'e region, the long-time behavior for the potential reads
\begin{align}
q(x,t)=(nt)^{-\frac{1}{n}}u_n(x(nt)^{-\frac{1}{n}})+\OO(t^{-\frac{3}{2n}}),\quad  t\rightarrow\infty,
\end{align}
where
$u_n$ solves the $n^{th}$ member of the Painlev\'e II hierarchy.

\end{theorem}
\subsection{Outline}

 In section \ref{section: ist and rhp}, we simply review the inverse scattering for the AKNS hierarchy. In  section \ref{section: overview of the strategies}, we summarize the idea of the $\dbar$-steepest descent method following \cite{dieng2018dispersive,chen2019longtime}. In the following sections  we first discuss the long-time behavior of the potential in the oscillating region. The general workflow is shown in Fig.\ref{fig:steps}. The first step (see section \ref{section: conjugation}) is so-called conjugation by which one can simultaneously factorize the jump matrix to lower/upper triangle and upper/lower triangle.  The next step (see section \ref{section: lenses opening}) is so called lenses-opening. In each interval where $\theta$ is monotonic, we can deform those intervals into new contours which are off the real line and the exponential terms will decay as $t$ goes to infinity on the new contours. The core idea of this step of the $\dbar$-steepest descent is to use Stokes' theorem to transfer contour integrals to double integrals, while in the original Deift-Zhou's method, this step is done by first performing rational approximation then analytic continuation. After lenses-opening, we will end up at a mixed $\dbar$-RHP. Next, from section \ref{section: Separate Contributions and phase reduction} and section \ref{section: model RHP}, we will first approximate the pure RHP. Three main steps of approximating the pure RHP are so-called localization, phase reduction and contribution separation, which lead to an exact solvable model RHP (also called Its' isomonodromy problem). Due to the exact solvability of the model RHP and the small norm theory, one can establish the existence and uniqueness of the pure RHP part of the mixed $\dbar$-RHP. The last step (see section \ref{section: dbar error}) is to estimate the errors by analyze the pure $\bar{\partial}$-problem which dominates the errors generated by approximating the pure RHP. Undo all the steps, we will eventually prove Theorem \ref{Theorem 1}. Then, in section \ref{section: fast decaying}, we will study the long-time behavior of the potential in the fast decaying region. Following similar analysis, we end up proving Theorem \ref{theorem fast decay}. The final section (see section \ref{section: painleve region}) is devoted to proving Theorem \ref{theorem fast decay}. In that section, we first give an algorithm to generate the Painlev\'e II hierarchy. Following the method of $\dbar$-steepest descent, we represent the long-time behavior of the potential by the solution to a member of the Painlev\'e II hierarchy.
\section{Inverse scattering transform and Riemann-Hilbert problem in $L^2$}
\label{section: ist and rhp}
In this section, we simply review the inverse scattering transform for the AKNS hierarchy in a certain weighted $L^2$ Sobolev space. For more details, we direct readers to Zhou's paper \cite{zhou98}.

The AKNS hierarchy is the integrable hierarchy associated with the following spectral problem:
\begin{align}
    \psi_x(x,t;z)=(-iz\sigma_3+Q(x,t))\psi(x,t;z),
    \label{akns: spectral problem}
\end{align}
where $\sigma_3=\begin{pmatrix}1 &0\\ 0&-1 \end{pmatrix}$ and $Q(x,t)=\begin{pmatrix}0&q(x,t)\\r(x,t)&0 \end{pmatrix}$. 

In the current paper, we only consider the defocusing type reduction:
\begin{align*}
    r(x,t)=q(x,t)\in \R,
\end{align*}
and we assume $q(x,t=0)\in H^{n-1,1}(\R,dx)$\footnote{This guarantees the time evolving of the initial data will stay in $H^{1,1}$. Roughly speaking, from Zhou's work, we know $q(x,0)\in H^{n-1,1}\subset H^{1,1}$ is mapped to $R(z)\in H^{1,n-1}$. Time evolution of the reflection coefficient gives $R(z)e^{itz^n}$, which belongs to $H^{1,1}$ due to the fact that $R(z)\in H^{1,n-1}$, and then the inverse scattering leads to $q(x,t)\in H^{1,1}$.}. For $t=0$, we are looking for solutions (so-called Jost solutions) of equation \eqref{akns: spectral problem} in $H^{1,1}(\R,dx)$, which satisfy the following boundary conditions at infinity:
\begin{align}
\psi_{\pm}=e^{ixz\sigma_3}+o(1),\quad x\rightarrow \pm \infty.
\end{align}
The scattering matrix $S(z)$ is then defined as
\begin{align}
    S(z):=\psi_+\psi_-^{-1}.
\end{align}
It is well known that $S$ enjoys the following properties: for $z\in \R$, 
\begin{align}
    S(z)=\begin{pmatrix}
    a(z) & \bar{b}(z)\\
    b(z) & \bar{a}(z)
    \end{pmatrix},
\end{align}
where $a,b$ can be represented in terms of the initial data and the eigenfunctions $\psi$. To find such representations, we consider
\begin{align}
    \mu^{(\pm)}=\psi_\pm e^{-izx\sigma_3}.
\end{align}
Then the spectral problem \eqref{akns: spectral problem} becomes:
\begin{align}
    (\mu^{(\pm)})_x=iz[\sigma_3,\mu^{(\pm)}]+Q\mu^{(\pm)}.
    \label{equation of mu}
\end{align}
Then the representations of $a,b$ read
\begin{align}
   a(z)&=\mu^{(+)}_{11}(x\rightarrow-\infty)=1-\int_{\R}q(y)\mu_{21}^{(+)}(y,z)dy,\label{a in terms of mu 11}\\
	b(z)&=\mu^{(+)}_{12}(x\rightarrow-\infty)=-\int_{\R}e^{2iyz}q(y)\mu_{22}^{(+)}(y,z)dy.\label{b in terms of mu 22}
\end{align}

From the above representations, it is straightforward to show that $a(z)=1+\OO(1/z)$ and $b(z)=\OO(1/z)$ as $z\rightarrow\infty$, and $a(z)$ can be analytically extended to the upper half plane.

Now, setting
\begin{align}
		m_+(x,z)&=(\mu_1^{(+)}(x,z)/a(z),\mu_2^{(-)}(x,z)),\ \Im{z}\geq 0,\\
		m_-(x,z)&=(\mu_1^{(-)}(x,z),\mu_2^{(+)}(x,z)/\bar{a}(z)),\ \Im{z}\leq 0,
\end{align}
we can then define the jump matrix on the real line by
\begin{align}
    v(z)=e^{-izx\text{ ad}\sigma_3}(m_-^{-1}m_+).
\end{align}
A direct computation shows
\begin{align}
    v(z)=\begin{pmatrix}
    1-|R(z)|^2 & -\bar{R}(z)\\
    R(z) & 1    
    \end{pmatrix},
    \label{jump matrix without time evolution}
\end{align}
where
\begin{align}
R(z)=\frac{b(z)}{a(z)}.
\label{reflection coeff}
\end{align}

The deformation of the spectral problem \eqref{akns: spectral problem} with respect to $t$  is governed by the following equation:
\begin{align}
    \psi_t(x,t;z)&=\left(\sum_{k=0}^nQ_k(x,t)z^k\right)\psi(x,t;z).\label{akns: time evolution}
\end{align}
To generate isospectral flow, $\psi$ need to satisfy the compatiblity condition, i.e., $\psi_{xt}=\psi_{tx}$. By this condition, one can uniquely determine $Q_k$ if the integration constants are assumed to be all zeros. One can systematically determines the $Q_k$'s via associated Lie algebra techniques, see for example \cite{MA2013117}. Through the Lie algebra, one can show the AKNS hierarchy is integrable, i.e, there are infinite many conservation laws. Moreover, using the powerful trace identity \cite{tu89}, one can easily show the bi-Hamiltonian structure of the AKNS hierarchy. Moreover, under the same framework, one can show that any linear combinations of the time-evolution problem are also integrable.

The compatibility condition of \eqref{akns: spectral problem} and \eqref{akns: time evolution} generates integrable PDEs, including the defocusing nonlinear Schr\"odinger equation, the modified KdV equation, the fifth-order modified KdV equation. Due to the decaying of the potential $Q$, it is easy to show the time evolution of the jump matrix $v$ is trivial. Formally speaking, since $S\psi_-=\psi_+$, taking derivatives with respect to $t$ on both sides leads to 
\begin{align*}
    S_t\psi_-+S\psi_{-,t}=\psi_{+,t},
\end{align*}
then by the time evolution equation on $\psi$, we have
\begin{align*}
    S_t\psi_-+S\left(\sum_{k=0}^nQ_k(x,t)z^k\right)\psi_-&=\left(\sum_{k=0}^nQ_k(x,t)z^k\right)\psi_+\\
    &=\left(\sum_{k=0}^nQ_k(x,t)z^k\right)S\psi_-,
\end{align*}
letting $x\rightarrow -\infty$, and since for the case of AKNS flows, all coefficients of $z^k,k=0,\cdots ,n-1,$ will vanish, we arrive at (see, e.g., \cite{ma2019application,ma2021}):
\begin{align*}
    S_t=[Q_nz^n,S].
\end{align*}
It is of our current interest that $Q_n=-icz^n\sigma_3$ for some positive constant $c$. Therefore, we have the time evolution for the scattering matrix
\begin{align}
    \label{time evolution scattering matrix}
    S(z;t)=e^{-icz^nt\ad{\sigma_3}}S(z),
\end{align}
where $e^{\ad{\sigma_3}}(\cdot):=e^{\sigma_3}(\cdot)e^{-\sigma_3}.$
This implies the time evolution of the jump matrix $v(z)$ (see \eqref{jump matrix without time evolution}), and we have
\begin{align}
    v_\theta(z):=e^{-it\theta(z;x,t)\ad{\sigma_3}}v(z),
    \label{jump matrix with time evolution}
\end{align}
where (in the case of the AKNS hierarchy) $\theta(z;x,t)=\frac{x}{t}z+cz^n$ for some positive constant $c$.

Finally, we formulate the direct scattering problem as a Riemann-Hilbert problem as follows:
\begin{problem}
\label{RHP m 0}
Looking for a 2 by 2 matrix-valued functions $m(z;x,t)$ such that
\begin{enumerate}[label={(\arabic*)}]
    \item $m(z)$ is analytic in $\C\backslash \R$;
    \item $m_+=m_-v_\theta,\quad z\in \R$;
    \item $m=I+m_1(x,t)/z+\OO(1/z^2)$, $z\rightarrow\infty$;
\end{enumerate}
where $v_\theta$ is defined in equation \eqref{jump matrix with time evolution} and  $m_{\pm}=\lim_{\epsilon\downarrow 0}m(z\pm i\epsilon)$.
\end{problem}

From the equation \eqref{equation of mu}, and the definition of the jump matrix $v$, we can recover 
the potential by
\begin{equation}
\begin{split}
    q(x,t)&=-2i\lim_{z\rightarrow\infty}[z(m-I)]_{12}\\
   \label{potential recover formula}  &=-2i (m_1(x,t))_{12}.
    \end{split}
\end{equation}
In the following sections, we will perform the $\dbar$-steepest descent method and study the asymptotic behavior for $t$ being sufficiently large.

\section{Overview of the strategies}
\label{section: overview of the strategies}

In this section, we will simply review the idea of Deift-Zhou's nonlinear steepest descent method and its variation, the $\dbar$-steepest descent method. In general, the key step in both methods is to deform the RHP. After the deformation, the new RHP is expected to be approximable locally as $t$ goes to $\infty$. Next, we will explain the main ideas of both methods. The notations in this section are used in this section only.

Let us consider the following RHP on $\R_+$:
\begin{align*}
    M_+(z)&=M_-(z)e^{-it\theta(z)\ad{\sigma_3}}V(z),\quad z\in \R_+,\\
    M(z)&\rightarrow I,\quad z\rightarrow \infty.
\end{align*}
 The main idea\footnote{A good summary of this method can be found in \cite{deift_its_zhou_1993}.} of Deift-Zhou's method is to find a factorization of $V(z)$,  say, $V(z)=V_-(z)V_+(z)$, such that $V_{\pm}(z)$ can be approximated by $\tilde{V}_{\pm}(z)$ which are analytic in the sectors $\Omega_+$ and $\Omega_-$ respectively, see Fig.\ref{idea of dbar steepest descent}.
\begin{figure}[h]
\centering
\begin{tikzpicture}
\draw (0,0)--(4,0);
\draw [fill] (0,0) circle [radius=0.05];
\node [below] at (0,0) {$O$};
\draw [line width=0.3 mm] (0,0)--(4,1.5);
\path [fill=gray,opacity=0.3] (0,0)--(3.8,0)--(3.8,3.8*3/8)--(0,0);
\node [above] at (3.8,1.5) {$\Sigma_{1}$};
\node [above] at (3,0.4) {$\Omega_+$};

\draw [line width=0.3 mm] (0,0)--(4,-1.5);
\path [fill=gray,opacity=0.6] (0,0)--(3.8,0)--(3.8,-3.8*3/8)--(0,0);
\node [below] at (3.8,-1.5) {$\Sigma_{2}$};
\node [below] at (3,-0.4) {$\Omega_-$};
\node [above] at (-1,1) {$\Omega$};

\end{tikzpicture}
\caption{Contour deformation}
\label{idea of dbar steepest descent}
\end{figure}
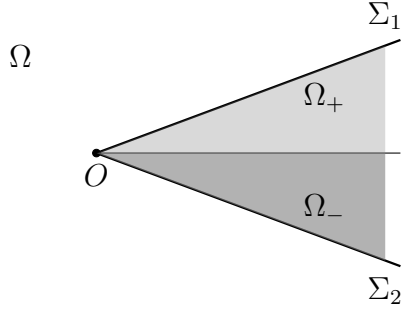
By introducing a new analytic function $\tilde{M}$ as follows:
\begin{align*}
    \tilde{M}|_{\Omega}&=M,\\
    \tilde{M}|_{\Omega_+}&=M\tilde{V}_+^{-1},\\
    \tilde{M}|_{\Omega_-}&=M\tilde{V}_-,
\end{align*}
we arrive at a new RHP:
\begin{align*}
    \tilde{M}_+&=\tilde{M}_-\tilde{V},\quad z\in \R_+\cup \Sigma_1\cup \Sigma_2,\\
    \tilde{V}&=\begin{cases}\tilde{V}_+,\quad z\in \Sigma_1,\\
    \tilde{V}_-,\quad z\in \Sigma_2,\\
    \tilde{V}_-^{-1}V_-V_+\tilde{V}_+^{-1},\quad z\in \R_+.
    \end{cases}
\end{align*}
Also, we want to guarantee that 
based on the signatures of the phase function $\Re{(i\theta(z))}$ in each sector, the new jumps converge rapidly to the identity away from $O$ as $t\rightarrow\infty$. Usually, one needs to deform the RHP several times. Eventually, the initial RHP can be approximated locally by the following fairly simple model RHP:
\begin{align*}
    M^{\sharp}_+&=M^{\sharp}_-e^{-it\tilde{\theta}(z)\ad{\sigma_3}}V(0),\quad z\in \R_+,\\
    M^{\sharp}&\rightarrow I,\quad z\rightarrow \infty,
\end{align*}
where $\tilde{\theta}(z)$ is a certain rational approximation to $\theta(z)$ near $z=0$. This model RHP can be solved explicitly and by undoing all deformations, one can track all errors in the middle steps. 

The Deift-Zhou method of analyzing errors is heavily based on the harmonic analysis for the Cauchy operators on contours, however, the $\dbar$-steepest descent method transfers the error estimation to some fairly simple estimations of certain double integrals. A natural way of connecting the contour integrals to the double integrals is to use Stokes' theorem (or the Cauchy-Green theorem): for any $C^1(\R^2\rightarrow \C)$ function $f(z):=f(z,\bar{z})$, we have
\begin{align*}
    \int_{\partial \Omega}f(z)dz=2i\int_{\Omega}\frac{\partial f(z)}{\partial \bar{z}}dxdy,
\end{align*}
where $z=x+iy$. So in the $\dbar$-steepest descent theory, 
we try to find an interpolation, say $E(z)$, between the old contour and the new one. Such an $E$ satisfies
\begin{align*}
    E(z)=\begin{cases}
    V_+(0),\quad z\in \Sigma_1^-,\\
    V_-(0),\quad z\in \Sigma_2^+,\\
    V_+(z),\quad z\in \R_+^+,\\
    V_-(z),\quad z\in \R_+^-,\\
    I,\quad z\in \Sigma_1^+\cup\Sigma_2^-,
    \end{cases}
\end{align*}
where all contours are orientated from $O$ to $\infty$ and $\Gamma^{\pm}$ mean the limit from left/right,
and it is $C^1$ in $\overline{\Omega_+\cup\Omega_-}$. Also, we want $e^{it\tilde{\theta}(z)\ad{\sigma_3}}V_{\pm}(0)$ go to 0 as $t\rightarrow\infty$. Now, let us set $\hat{M}=M(z)E(z)$, we obtain the so-called $\dbar$-RHP:
\begin{enumerate}
    \item (The RHP) $\hat{M}_+=\hat{M}_-e^{-it{\theta}(z)\ad{\sigma_3}}\hat{V}(z),\quad z\in \Sigma_1\cup\Sigma_2$, where
    \begin{align*}
        \hat{V}(z)=\begin{cases}
        e^{-it{\theta}(z)\ad{\sigma_3}}V_+(0),\quad z\in \Sigma_1,\\
        e^{-it{\theta}(z)\ad{\sigma_3}}V_-(0),\quad z\in \Sigma_2.
        \end{cases}
    \end{align*}
    \item (The $\dbar$-problem) For any $z\in \C$, we have
    \begin{align*}
        \dbar \hat{M}=\hat{M}E^{-1}\dbar{E}.
    \end{align*}
\end{enumerate}
The deformation of the RHP follows from Deift-Zhou's method, but the error estimations here are transferred to a dbar problem, which turns out to be equivalent to some singular integral equation with respect to the area measure. Then through some fairly simple estimates on the double integrals, one will obtain the same error estimates as 
the Deift-Zhou method. In the following sections, we will apply the $\dbar$-steepest descent to the 
defocusing mKdV type reduction of the AKNS hierarchy.

\section{Conjugation}
\label{section: conjugation}

In this section, we will factorize the jump matrix (as defined by equation \eqref{jump matrix} ) in a way that it can be used for deforming the RHP. It is easy to see that the jump matrix enjoys the following two kinds of factorization:
\begin{align}
    v_\theta(z)=\begin{cases}
    \begin{pmatrix}
    1 & -\bar{R}(z)e^{-2it\theta}\\
    0& 1
    \end{pmatrix}\begin{pmatrix}
    1 & 0\\
    R(z)e^{2it\theta} & 1
    \end{pmatrix},\\
    \begin{pmatrix}
    1 & 0\\
    \frac{R(z)}{1-|R(z)|^2}e^{2it\theta}& 1
    \end{pmatrix}(1-|R|^2)^{\sigma_3}\begin{pmatrix}
    1 & -\frac{\bar{R}(z)}{1-|R|^2}e^{-2it\theta}\\
    0 & 1
    \end{pmatrix}.
    \end{cases}
\end{align}

In the light of the main ideas we described in the last section, we want to remove the middle term in the second factorization. By doing so, we can eventually find the proper factorization based on the signatures of the  $\Re(i\theta)$. Due to our assumptions on $\theta$, near a stationary phase point (say $|z-z_j|<\epsilon$, for some small positive $\epsilon$), $\theta=\theta(z_j)+\frac{\theta''(z_j)}{2}(z-z_j)^2+\OO(|z-z_j|^3)$. If $\theta''(z_j)>0$, then $\Re(i\theta(z))$ is negative in the line (I): $\{z=z_j+re^{i\alpha}, r\in (-\epsilon,\epsilon) \text{ with fixed } \alpha\in (0,\pi/2)\}$, and it is positive in the line (II): $\{z=z_j+re^{i\alpha}, r\in (-\epsilon,\epsilon) \text{ with fixed }\alpha\in (-\pi/2,0)\}$.  On the line (I), notice that $e^{2it\theta}$ decays to 0 as $t\rightarrow \infty$,  we can deform the jump on the contour right to the stationary phase point using the first factorization. With the same argument on the line (II), we can deform the jump on the contour left to the stationary phase point using the second factorization. If $\theta''(z_j)<0$, notice now $e^{2it\theta}$ decays to 0 as $t\rightarrow\infty$ on the line (II), and 
thus we need the second factorization for the jump on the contour right to the stationary phase point and the first factorization for the jump on the contour left to the stationary phase point. Motivated by the above arguments, we denote $D_\pm=\{z\in \R: \pm \theta'(z)>0\}$

To eliminate the diagonal matrix in the second factorization, we introduce a scalar RHP:
\begin{equation}
\label{scalar RHP delta}
\begin{split} 
    \delta_+&=\delta_-[(1-|R|^2)\chi_{D_-}+\chi_{D_+}],\quad z\in \R,\\
    \delta(z)&=1+\OO(z^{-1}),\quad z\rightarrow\infty.
    \end{split}
\end{equation}
Then by conjugating the initial RHP, we arrive at a new RHP: 
\begin{problem} Looking for a 2 by 2 matrix-valued function $m^{[1]}(z;x,t)$ such that
\label{RHP m 1}
\begin{enumerate}[label={(\arabic*)}]
\item $m^{[1]}_+=m^{[1]}_-\delta_-^{\sigma_3}v_\theta\delta_+^{-\sigma_3},\quad z\in \R;$
\item $m^{[1]}=I+\OO(z^{-1}),\quad z\rightarrow\infty.$
\end{enumerate}

By denoting $v_\theta^{[1]}:=\delta_-^{\sigma_3}v_\theta\delta_+^{-\sigma_3}$, the new jump matrix reads
\begin{equation}
    v_\theta^{[1]}(z)=\begin{cases}
    \begin{pmatrix}
    1 & -\bar{R}(z)\delta^2(z)e^{-2it\theta}\\
    0 & 1
    \end{pmatrix}
    \begin{pmatrix}
    1 & 0\\
    R(z)\delta^{-2}e^{2it\theta} & 1
    \end{pmatrix},\quad z\in D_+,\\
    \begin{pmatrix}
    1 & 0\\
    \frac{R(z)\delta^{-2}_-e^{2it\theta}}{1-|R|^2} & 1
    \end{pmatrix}\begin{pmatrix}
    1 & -\frac{\bar{R}(z)\delta_+^2(z)e^{-2it\theta}}{1-|R|^2}\\
    0 & 1
    \end{pmatrix},\quad z\in D_-.
    \end{cases}
\end{equation}
\end{problem}

The scalar RHP \eqref{scalar RHP delta} has been carefully studied in the literature (see for example \cite{beals_deift_tomei_1988} Lemma 23.3, \cite{DZ03}, \cite{varzugin96} and \cite{do2009nonlinear}). Here we just list some of the properties. First, the solution to the RHP \eqref{scalar RHP delta} can be represented as follows:
\begin{equation}
    \log{(\delta(z))}=(C_{D_-}(\log(1-|R|^2)))(z),z\in \C\backslash D_-,
\end{equation}
where the Cauchy operator $(C_{D_-}f)(z)=\frac{1}{2\pi i}\int_{D_-}\frac{f(s)}{s-z}ds$ . Since we assume $R(z)\in H^{1,1}_1(\R,dz)=H^{1,1}\cap \{f:|f|<1\}$, one can show $\log(1-|R|^2)$ is in $H^{1,0}$, and 
then by the Sobolev embedding, we know it is also H\"older continuous with index $1/2$. Then, by the 
Privalov-Plemelj theorem, which says that Cauchy operator perseveres H\"older continuity with index less than 1, one can show $\log(\delta(z))$ is H\"older continuous with index $1/2$ except for the end points. Next we study the behavior near those points. 

Let us denote
\begin{align}
\eta(z)=-\frac{1}{2\pi}\log(1-|R(z)|^2), \quad z\in \R.
\label{definite of eta}
\end{align} 
We will prove the following proposition.

First we define a tent function supported on the interval $[-\epsilon,\epsilon]$,
\begin{align}
s_\epsilon(z)=\begin{cases}
0,\quad |z|\geq \epsilon\\
-\frac{1}{\epsilon}z+1,\quad 0< z<\epsilon,\\
\frac{1}{\epsilon}z+1,\quad -\epsilon< z\leq 0.
\end{cases}
\end{align}

\begin{prop}\label{delta decomposition}
For each $\epsilon>0$, and $\epsilon\leq \frac{1}{3}\min_{j\neq k}|z_j-z_k|, $ there exists a neighborhood $I=I(\epsilon)$ such that the identity
\begin{align*}
\log(\delta(z))&=i\int_{D_-\backslash I}\frac{\eta(s)}{s-z}ds+i\sum_{j=1}^l\left[\eta(z_j)(1+\log(z-z_j))\right]\varepsilon_j\\
&+i\sum_{j=1}^l{\int_{I\cap D_-}\frac{\eta(s)-\eta_j(s)}{s-z}ds}\\
&+i\sum_{j=1}^l\frac{1}{\epsilon}\eta(z_j)[(z-z_j)\log(z-z_j)-(z-z_j+\varepsilon_j\epsilon)\log(z-z_j+\varepsilon_j\epsilon)]
\end{align*}
is true, where $\varepsilon_j=\text{sgn}(\theta''(z_j))$,  $\eta_j(z)=\eta(z_j)s_\epsilon(z-z_j)$ and see \eqref{definite of eta} for the definition of $\eta$. As for the logarithm function, the branch is chosen such that $\text{argument}\in (-\pi,\pi)$.
\end{prop}
\begin{proof}
Let $I=\cup_{j=1}^l(I_{j+}\cup I_{j-})$, where $I_{j\pm}=\{z:0<\pm(z-z_j)<\epsilon\}$. Now we have
\begin{align*}
\log(\delta(z))&=i\int_{D_-\backslash I}\frac{\eta(s)}{s-z}ds\\
&+i\sum_{j=1}^l(\int_{I_{j-}\cap D_-}+\int_{I_{j+}\cap D_-}\frac{\eta(s)}{s-z}ds).
\end{align*}
For each $j$, we have
\begin{align*}
\int_{I_{j-}}\frac{\eta(s)}{s-z}ds=\int_{I_{j-}}\frac{\eta(s)-\eta_j(s)}{s-z}ds+\int_{I_{j-}}\frac{\eta_j(s)}{s-z}ds.
\end{align*}
The first integral on the right hand side is the non-tangential limit as $z\rightarrow z_j$ and the second one generates a logarithm singularity near $z_j$. In fact, direct computation shows
\begin{align*}
\int_{I_{j-}}\frac{\eta_j(s)}{s-z}ds&=\eta(z_j)+\frac{1}{\epsilon}[(z-z_j)\log(z-z_j)-(z-z_j+\epsilon)\log(z-z_j+\epsilon)]\eta(z_j)\\
&+\eta(z_j)\log(z-z_j).
\end{align*}
Similarly, for $I_{j+}$,
\begin{align*}
\int_{I_{j+}}\frac{\eta_j(s)}{s-z}ds&=-\eta(z_j)+\frac{1}{\epsilon}[(z-z_j)\log(z-z_j)-(z-z_j-\epsilon)\log(z-z_j-\epsilon)]\\
&-\eta(z_j)\log(z-z_j).
\end{align*}
And note that only one of the $I_{j\pm}\cap D_-$ is nonempty, which depends on the sign of the second derivative of the phase function $\theta$. By assembling all together, the proof is done.
\end{proof}
\begin{remark}
The proposition tells us how the function $\delta(z)$ behavior near the saddle points. Near the saddle points $z_j$, $\delta(z)$ has a mild singularity $(z-z_j)^{i\eta(z_j)}$. Fortunately, those singularities are bounded along any ray off $\R$ and hence in some sense they do not affect asymptotics much. It is worth mentioning that one can ignore the mild singularity by introducing an auxiliary function, see Lemma 3.1 in \cite{dieng2018dispersive}.
\end{remark}

\section{Lenses opening}
\label{section: lenses opening}
The purpose of lens-opening is to deform the RHP on the real line to a new RHP on new contours 
such that jumps on the new contours will rapidly decay to $I$ as $t\rightarrow\infty$. 
We first study the signature of $\Im{\theta}$ near the saddle point $z_j$.
\begin{figure}[h]
\centering
\begin{tikzpicture}
\draw (-4,0)--(4,0);
\node [above] at (-1,2) {$\theta'=0$};
\draw [->](-1,2)--(0,0.1);
\draw [fill] (0,0) circle [radius=0.05];
\node [below] at (0,0) {$z_j$};
\node [below] at (2,0) {$I_{j+}$};
\node [below] at (-2,0) {$I_{j-}$};
\draw [line width=0.3 mm] (0,0)--(4,1.5);
\path [fill=gray,opacity=0.3] (0,0)--(3.8,0)--(3.8,3.8*3/8)--(0,0);
\draw (0.5,0) arc (0:27:0.4) ;
\node [right] at (0.5,0.1) { $\alpha$};
\node [above] at (3.8,1.5) {$\Sigma_{j,1}$};

\draw [line width=0.3 mm] (0,0)--(-4,1.5);
\path [fill=gray,opacity=0.6] (0,0)--(-3.8,0)--(-3.8,3.8*3/8)--(0,0);
\draw (-0.5,0) arc (180:180-27:0.4) ;
\node [above] at (-3.8,1.5) {$\Sigma_{j,2}$};

\draw [line width=0.3 mm] (0,0)--(-4,-1.5);
\path [fill=gray,opacity=0.3] (0,0)--(-3.8,0)--(-3.8,-3.8*3/8)--(0,0);
\draw (-0.5,0) arc (180:180+27:0.4) ;
\node [below] at (-3.8,-1.5) {$\Sigma_{j,3}$};

\draw [line width=0.3 mm] (0,0)--(4,-1.5);
\path [fill=gray,opacity=0.6] (0,0)--(3.8,0)--(3.8,-3.8*3/8)--(0,0);
\node [below] at (3.8,-1.5) {$\Sigma_{j,4}$};

\node at (3,0.65) {$\Omega_{j,1}$} ;
\node at (-3,0.65) {$\Omega_{j,3}$} ;
\node at (-3,-0.65) {$\Omega_{j,4}$} ;
\node at (3,-0.65) {$\Omega_{j,6}$} ;
\node at (0,0.95) {$\Omega_{j,2}$} ;
\node at (0,-0.95) {$\Omega_{j,5}$} ;
\end{tikzpicture}
\caption{Notations for studying signatures of $\Im(\theta(z))$ near $z_j$}
\label{notation for Im theta}
\end{figure}
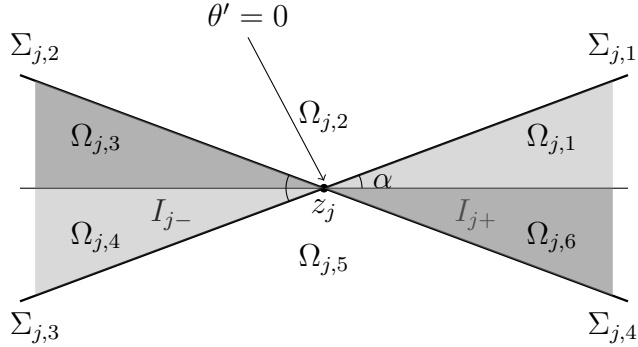
Let us denote $I_{j+}=[z_j,\frac{z_j+z_{j+1}}{2}]$ and $I_{j-}=[\frac{z_j+z_{j-1}}{2},z_j]$. Two cases need to be discussed. The first case is $\theta''(z_j)>0$, and so we have $I_{j\pm}\subset D_{\pm}$. The second case
is $\theta''(z_j)<0$, and then we have $I_{j\pm}\subset D_{\mp}$. 

Recall the factorization of the conjugated jump matrix $v^{[1]}_\theta$, to deform it from $I_{j+}$ to $\Sigma_{j,1}$, we need make sure the exponential term $e^{2it\theta(z)}$ decays rapidly to $I$ on $\Sigma_{j,1}$, and thus we need to discuss $\Re(i\theta)$ on $\Sigma_{j,1}$. Considering a Taylor approximation of $\theta(z)$ near $z_j$, we have $\theta(z)=\theta(z_j)+\varepsilon_jA_j(z-z_j)^2+O(|z-z_j|^3)$, where $A_j=\left|\frac{\theta''(z_j)}{2}\right|$ and $\varepsilon_j=\text{sgn}\{\theta''(z_j)\}$. 

Let $z=z_j+u+iv=z_j+\rho e^{i\phi}$. Then $\Im(\theta(z))=\varepsilon_j A_j\rho ^2 \sin(2\phi)+O(\rho^3)$, where $\phi\in (0,\alpha]$ is fixed. Now we define the regions $\Omega_{j,n},\ n=1,\cdots ,6,$ as follows:
\begin{equation}
\label{ definition of the regions}
\begin{split}
    \Omega_{j,1}&=\left\{z=z_j+\rho e^{i\phi},\phi\in (0,\alpha],\rho\in (0,\frac{|z_j-z_{j+1}|}{2\cos{\alpha}}),\Re z\in I^{\varepsilon_j}_{j+}\right\},\\
    \Omega_{j,3}&=\left\{z=z_j+\rho e^{i\phi},\phi\in [\pi-\alpha,\pi),\rho\in (0,\frac{|z_j-z_{j-1}|}{2\cos{\alpha}}),\Re 
    z\in I^{\varepsilon_j}_{j-}\right\},\\
    \Omega_{j,2}&=\C^+\backslash (\Omega_{j,1}\cup\Omega_{j,3}),\\
    \Omega_{j,4}&=\left\{z=z_j+\rho e^{i\phi},\phi\in (\pi,\pi+\alpha],\rho\in (0,\frac{|z_j-z_{j-1}|}{2\cos{\alpha}}),\Re z\in I^{\varepsilon_j}_{j-}\right\},\\
    \Omega_{j,6}&=\left\{z=z_j+\rho e^{i\phi},\phi\in [-\alpha,0),\rho\in (0,\frac{|z_j-z_{j+1}|}{2\cos{\alpha}}),\Re z\in I^{\varepsilon_j}_{j+}\right\},\\
    \Omega_{j,5}&=\C^-\backslash (\Omega_{j,4}\cup\Omega_{j,6}),
    \end{split}
\end{equation}
where
\begin{align*}
    I^{\varepsilon_j}_{j\pm}=\begin{cases}
    I_{j\pm},\quad \varepsilon_j=1,\\
    I_{j\mp},\quad \varepsilon_j=-1.
    \end{cases}
\end{align*}
Since the number of real saddle points is finite, we can always choose a sufficiently small $\alpha$, such that for each $j$, $e^{2it\theta}$ decays to 0 in $\Omega_{j,1}\cup\Omega_{j,4}$ and $e^{-2it\theta}$ decays to 0 in $\Omega_{j,3}\cup\Omega_{j,6}$.

Now we are in the position to open the lenses. First we introduce a bounded smooth function $\mathcal{K}$ defined on $[0,2\pi]$ such that
\begin{equation}
\begin{split}
    \mathcal{K}(0)=1,\\ 
    \mathcal{K}(\alpha)=0,\\
    \text{Period of }\mathcal{K} \text{ is $\pi$},\\
    \mathcal{K} \text{ is even function}.
    \end{split}
    \label{prop of K}
\end{equation}
 Consider $\varepsilon_j=1$ first. Then the  $\bar{\partial}$ extension functions are as follows. Let $z-z_j=u+iv=\rho e^{i\phi}$,and 
for the case $\varepsilon_j=1$, we set
\begin{equation}
\label{definitions of E j k}
\begin{split}
E_{j,1}(z)&=\K(\phi)R(u+z_j)\delta^{-2}(z)\\
&+[1-\K(\phi)]R(z_j)\delta_j^{-2}(z-z_j)^{-2i\varepsilon_j\eta(z_j)},\quad z\in \Omega_{j,1},\\
E_{j,3}(z)&=\K(\pi-\phi)(-\frac{\bar{R}(u+z_j)}{1-|R(u+z_j)|^2}\delta_+^2(z))\\
&+[1-\K(\pi-\phi)](-\frac{\bar{R}(z_j)}{1-|R(z_j)|^2}\delta_{j}^2(z-z_j)^{2i\varepsilon_j\eta(z_j)}),\quad z\in \Omega_{j,3},\\
E_{j,4}(z)&=\K(\pi+\phi)(\frac{R(z_j+u)}{1-|R(z_j+u)|^2}\delta_-^{-2}(z))\\
&+[1-\K(\pi +\phi)](\frac{R(z_j)}{1-|R(z_j)|^2}\delta_j^{-2}(z-z_j)^{-2i\varepsilon_j\eta(z_j)}),\quad z\in \Omega_{j,4},\\
E_{j,6}(z)&=\K(-\phi)(-\bar{R}(z_j+u)\delta^2(z))\\
&+[1-\K(-\phi)](-\bar{R}(z_j)\delta_j^2(z-z_j)^{2i\varepsilon_j\eta(z_j)}),\quad z\in \Omega_{j,6},\\
\end{split}
\end{equation}
where
 $$\delta_j=\lim_{\substack{z=z_j+\rho e^{i\phi},\\ \rho\rightarrow 0,\\ \phi\in (0,\pi/2)}}\delta(z)(z-z_j)^{i\eta(z_j)}.$$
 
For the case $\varepsilon_j=-1$, one only needs to switch the index $1$ with $3$ and $4$ with $6$. For the sake of simplicity, in what follows, we focus just on the case $\varepsilon_j=1$. The extension functions can be considered as interpolations between jumps on the old and new contours. Using the extension functions $E_{j,k}, k=1,3,4,6$, we can construct the lens-opening matrices $O(z)$ as follows:
\begin{align}
O(z)=\begin{cases}
O_{j,n}(z)=\begin{pmatrix}
1 & 0\\
(-1)^nE_{j,n}e^{2it\theta(z)}& 1
\end{pmatrix},\quad z\in \Omega_{j,n},\quad n=1,4,\\
O_{j,m}(z)=\begin{pmatrix}
1 & (-1)^{m}E_{j,m}e^{-2it\theta(z)}\\
0 & 1
\end{pmatrix},\quad z\in \Omega_{j,m},\quad m=3,6,\\
O_{j,k}(z)=I,\quad z\in \Omega_{j,k},\quad k=2,5.
\end{cases}
\label{def of open lense}
\end{align} 
Then lens-opening is performed by multiplying $O(z)$ to the right of the matrix $m^{[1]}$. 
Let us denote $m^{[2]}(z)=m^{[1]}(z)O(z),z\in \C\backslash \R$. 
Due to the lacking of analyticity of $O(z)$ (in fact, since we only assume $R(z)\in C^1(\R)$, $O(z)$ is also just in $C^1(\R^2)$\footnote{Here, $R(z)\in C^1(\R)$  means $R(z)$ is a function defined on the real line with continuous first order derivative. While since $O(z)$ is a matrix-valued function defined on the complex plan, so $O(z)\in C^1(\R^2)$ means all the entries have continuous first-order derivatives with respect to $z$ and $\bar{z}$.}), we arrive at the following mixed $\dbar$-Riemann-Hilbert problem($\dbar$-RHP):
\begin{pproblem} Looking for a 2 by 2 matrix-valued function $m^{[2]}$ such that
\label{dbar RHP}
\begin{enumerate}[label={(\arabic*)}]
\item{The RHP}
\begin{itemize}
    \item $m^{[2]}(z)\in C^1(\R^2\Sigma)$;
    \item $m^{[2]}_+=m^{[2]}_-v^{[2]}_\theta,\quad z\in \cup_{j=1,...,l,k=1,2,3,4}\Sigma_{j,k},$
    where the jump matrices read
    \begin{align}
    \label{dbar RHP 1}
        v^{[2]}_\theta=\begin{cases}
        O_{j,1}^{-1}, \quad z\in \Sigma_{j,1},\\
        O_{j,3}^{-1}, \quad z\in \Sigma_{j,2},\\
        O_{j,4}, \quad z\in \Sigma_{j,3},\\
        O_{j,6}, \quad z\in \Sigma_{j,4};
        \end{cases}
    \end{align}
    \item $m^{[2]}(z)=I+\OO(z^{-1}),\quad z\rightarrow\infty$.
\end{itemize}

\item{The $\dbar$-problem}

For $z\in \C$, we have
\begin{align}
\label{dbar RHP 2}
\dbar m^{[2]}(z)=m^{[2]}(z)\dbar O(z).
\end{align}
\end{enumerate}
\end{pproblem}

To close this section, we state a bound estimate for $\dbar E_{j,k}$, which will be used in later sections.
\begin{lemma}\label{dbar essential estimate}
For $j=1\cdots l,\, k=1,2,3,4$, and $z\in \Omega_{j,k},u=\Re(z-z_j)$,
\begin{align}
|\dbar E_{j,k}(z)|\leq c(|z-z_j|^{-1/2}+|R'(u+z_j)|).
\end{align}
\end{lemma}
\begin{proof}
In the polar coordinates, $\dbar=\frac{e^{i\phi}}{2}(\partial_\rho+i\rho^{-1}\partial_{\phi})$. For $z$ in any ray starting from $z_j$ and off the real line, we have
\begin{align*}
\dbar E_{j,1}(z)&=\frac{ie^{i\phi}\K'(\phi)}{2\rho}[R(u+z_j)\delta^{-2}(z)-R(z_j)\delta^{-2}_j(z-z_j)^{-2i\eta(z_j)}]\\
&+\K(\phi)R'(u+z_j)\delta^{-2}(z).
\end{align*}

From Proposition \ref{delta decomposition}, we know $|\delta(z)-\delta_j(z-z_j)^{i\eta(z_j)}|\leq c|z-z_j|^{1/2}$. Also since
\begin{align*}
    \delta(z)^{-1}=e^{-C_{D_-}(\log{(1-|R|^2)})},
\end{align*}
it is evident that $\delta(z)^{-1}$ is bounded. Therefore
\begin{align*}
    |\delta^{-2}(z)-\delta^{-2}_j(z-z_j)^{-2i\eta(z_j)}|\leq c|z-z_j|^{1/2}.
\end{align*}
And we have\footnote{In the middle steps, $c$ means a generic positive constant.} 
\begin{align*}
   | R(u+z_j)&\delta^{-2}(z)-R(z_j)\delta^{-2}_j(z-z_j)^{-2i\eta(z_j)}|\\
   &\leq |R(u+z_j)-R(z_j)||\delta^{-2}(z)|\\
   &+|\delta^{-2}(z)-\delta^{-2}_j(z-z_j)^{-2i\eta(z_j)}||R(z_j)|\\
   &\leq c|\int_{z_j}^{u+z_j}R'(s)ds|+c|z-z_j|^{1/2}\\
   & \text{by Cauchy-Schwartz inequality}\\
   & \leq c\|R'\|_{L^2}|z-z_j|^{1/2}+c|z-z_j|^{1/2}\\
   & \leq c|z-z_j|^{1/2}.
\end{align*}
Therefore
\begin{equation}
\begin{split}
 	|\dbar E_{j,1}(z)|&\leq c\rho^{-1}|z-z_j|^{1/2}+c|R'(u+z_j)|\\
 	&\leq c(|z-z_j|^{-1/2}+|R'(u+z_j)|).
 	\end{split}
 \end{equation} 
 Here we have use the fact that $u\leq \rho$, which implies $|z-z_j|^{1/2}/\rho=u^{1/2}/\rho\leq u^{-1/2}$.
Noting also that $\sup |R|<1$, we have $\frac{R}{1-|R|^2}\leq \frac{R}{1-\sup |R|}$, and thus all the estimates for $E_{j,1}$ can be smoothly moved to $E_{j,k},k=3,4,6$.
\end{proof}

\section{Separate contributions and phase reduction}
\label{section: Separate Contributions and phase reduction}
 The RHP and the mixed $\dbar$-RHP we have discussed above are global. In this section, we shall approximate the global RHP by performing two steps: (1) separate contributions from each stationary phase point, (2) phase reduction. Before that, let us first consider two saddle points $z_j,z_{j+1}$, and discuss $\varepsilon_j=1=-\varepsilon_{j+1}$ for example.  We will first remove the vertical segments, see Fig.\eqref{sigma j 1/2}: $$\Sigma_{j+\frac{1}{2}}:=\Omega_{j,1}\cap \Omega_{j+1,3} \cup \Omega_{j,6}\cap \Omega_{j+1,4}\backslash\R,$$
 where  $\Omega_{j,\cdot}$'s are defined in \ref{ definition of the regions}.
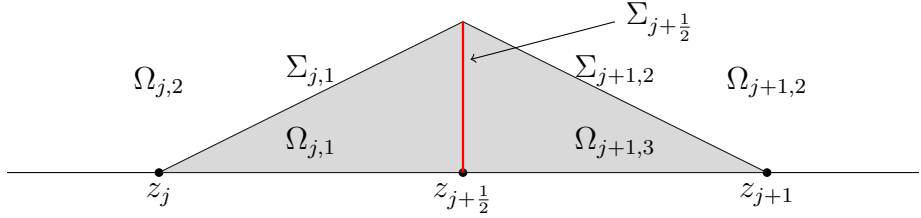
\begin{figure}[h]
\centering
\begin{tikzpicture}
\draw (-6,0)--(6,0);
\draw [fill] (-4,0) circle [radius=0.05] node [below] {$z_j$};
\draw [fill] (4,0) circle [radius=0.05] node [below] {$z_{j+1}$};
\draw [fill] (0,0) circle [radius=0.05] node [below] {$z_{j+\frac{1}{2}}$};
\draw (-4,0)--(0,2)--(4,0);
\path [fill=gray,opacity=0.3] (-4,0)--(0,2)--(4,0)--(-4,0);
\draw [red,line width=0.3 mm] (0,2)--(0,0);
\draw [->] (2,2)--(0.1,1.5) ;
\node [right] at (2,2) {$\Sigma_{j+\frac{1}{2}}$};
\node at (-2,0.4) {$\Omega_{j,1}$};
\node at (2,0.4) {$\Omega_{j+1,3}$};
\node at (-4,1.2) {$\Omega_{j,2}$};
\node at (4,1.2) {$\Omega_{j+1,2}$};
\node [above] at (-2,1){$\Sigma_{j,1}$};
\node [above] at (2,1){$\Sigma_{j+1,2}$};
\end{tikzpicture}
\caption{Jumps in a small triangular region.}
\label{sigma j 1/2}
\end{figure}

Recall the constructions of $E_{j,1}$ and $E_{j+1,3}$ (see \eqref{definitions of E j k}), 
the boundary value of $m^{[2]}(z)$ on $\Sigma_{j+\frac{1}{2}}$ from $\Omega_{j,1}$ is 
\begin{equation*}
m^{[1]}(z_{j+1/2}+iv)O_{j,1}(z_{j+1/2}+iv),
\end{equation*}
while from $\Omega_{j+1,3}$ it is
\[m^{[1]}(z_{j+1/2}+iv)O_{j+1,3}(z_{j+1/2}+iv).\]
Both correspond to locally increasing parts of the phase function, and thus correspond to an upper/lower factorization. So the jump on the new contour $\Sigma_{j+1/2}$ is $O_{j+1,3}O^{-1}_{j,1}(z)$, $z=z_{j+\frac{1}{2}}+iv$, where the nontrivial entry is (regarding the property of $\mathcal{K}$ and definitions of those matrix $O_{j,k}$, see \eqref{prop of K} and \eqref{def of open lense} ):
\begin{align*}
  (1-\K(\phi))&[R(z_j)\delta_j^{-2}(z_{j+1/2}-z_j+iv)^{-2i\eta(z_j)}\\
  &-R(z_{j+1})\delta_{j+1}^{-2}(z_{j+1/2}-z_{j+1}+iv)^{-2i\eta(z_{j+1})}]e^{2it\theta(z_{j+1/2}+iv)},
\end{align*}
	
with $v\in (0,(z_{j+1/2}-z_j)\tan(\alpha))$ and $\phi=\arg{(z-z_j)}$. 

Note that 
\begin{align*}
|(z_{j+1/2}-z_j+iv)^{-2i\eta(z_j)}|&=e^{2\eta(z_j)\phi}\leq e^{2\eta(z_j)\alpha}.
\end{align*}
and
\[
	|e^{2it\theta(z_{j+1/2}+iv)}|\leq ce^{-2tdv},\quad d=(z_{j+1}-z_{j})/2.
\]
Thus we have, for any $z\in \Sigma_{j+\frac{1}{2}}$,
\begin{align*}
O_{j+1,3}O^{-1}_{j,1}-I=
\mathcal{O}(e^{-ct}),\quad t\rightarrow \infty,
\end{align*}
where $c$ is some generic positive constant. Since the jump is close to $I$, by a small norm theory, the solution will also be close to $I$. In fact, we have the following estimate for the potential
\begin{align*}
    |\lim_{z\rightarrow\infty}z&(m^{[2]}|_{\Sigma_{j+\frac{1}{2}}}-I)|\\
    & \leq \frac{1}{2\pi}\int_0^{d\tan{(\alpha)}}\left|m^{[2]}_-(z_{j+1/2}+is)\right|\left|O_{j+1,3}O_{j,1}^{-1}(z_{j+1/2}+is)-I\right|\text{ds}\\
    &\leq \frac{1}{2\pi}\int_0^{d\tan{(\alpha)}}\left|m^{[2]}_-(z_{j+1/2}+is)\right|e^{-2tsd}\text{ds}\\
    &\leq \frac{1}{2\pi}\|m^{[1]}|_{\Sigma_{j+\frac{1}{2} ,3}}\|_\infty\|O_{j+1,3}\|_\infty\int_{0}^{d\tan{(\alpha)}}e^{-2tsd}\text{ds}\\
    &=\OO(t^{-1}),
\end{align*}
where we assume $m^{[1]}$, as a solution to the conjugated RHP, exists\footnote{The existence and uniqueness will be discussed later.}. So it is analytic in a neighborhood of $\Sigma_{j+\frac{1}{2}}$ and hence it is bounded on $\Sigma_{j+\frac{1}{2}}$. By the definition (see \eqref{definitions of E j k}) of $O_{j+1,3}$, it is continuous in $\Sigma_{j+\frac{1}{2}}$ and does not blow up at the endpoints of $\Sigma_{j+\frac{1}{2}}$. So $\|O_{j+1,3}\|_\infty$ is also finite\footnote{Here the $L^\infty(\Sigma)$ norm $\|f(z)\|_{\infty}$ 
means $\sup_{z\in\Sigma}|f(z)|$, where $|f(z)|=\max_{i,j=1,2,z\in \Sigma}|f_{i,j}(z)|$.}. Therefore, we can remove all those vertical segments by paying a price of error $\OO(t^{-1})$, which will be dominated by the error generated by the $\dbar$-problem (it is $\OO(t^{-3/4})$, we will show it in a moment.) Let us denote the new RHP by $\tilde{m}^{[2]}$. To make it clear, we note that the jumps for $\tilde{m}^{[2]}$ are
\begin{align*}
    \tilde{v}^{[2]}(z)=\begin{cases}
    v^{[2]}(z),\quad z\in \cup_{j=1,..,l,k=1,2,3,4}\Sigma_{j,k},\\
    I,\quad z\in \cup_{j=1,..,l}\Sigma_{j+\frac{1}{2}}\cup\R.
    \end{cases}
\end{align*}

Next, we will show that the RHP for $\tilde{m}^{[2]}$ can be localized to each saddle point. For example, near $z_j$, along the segment  $\Sigma_{j,1}:z=z_j+u+iv, \arg{z}=\alpha$, we have 
\begin{align*}
    |E_{j,1}e^{2it\theta}|\leq ce^{-2t\tan(\alpha)u^2}
\end{align*}

    It is well-known \cite{DZ93,do2009nonlinear} that the $|E_{j,1}e^{2it\theta}|\leq ce^{-2t\tan(\alpha)u^2}$, where let $u\geq u_0>0$, and then the jump matrix will go to $I$ with decaying rate at $\mathcal{O}(e^{-ct}),c>0$, as $t\rightarrow\infty$. The RHP is localized in the small neighborhoods of those stationary phase points. Note  that near each $z_j$, we have
    \begin{align*}
    \theta(z)=\theta(z_0)+\frac{\theta''(z_0)(z-z_0)^2}{2}+\OO{(|z-z_j|^3)}.
    \end{align*}
    By a similar argument of Lemma 3.35 in \cite{deift_its_zhou_1993} or subsection 8.2 
    in \cite{do2009nonlinear} for the phase reduction, the error generated by reducing the phase function $\theta$ to $\theta(z_0)+\frac{\theta''(z_0)(z-z_0)^2}{2}$ will be bounded by $\OO{(t^{-1})}$. Both analysis of the  mentioned references  are based on the analysis of the so-called Beals-Coifman operator \cite{BC84}. Now we shall simply describe it here. For the sake of simplicity, we only consider the RHP on the contour $\Sigma_{j,1}$ (for more details, we direct the interested reader to \cite{DZ03}):
    \begin{problem} Looking for 2 by 2 matrix-valued function $\tilde{m}^{[2]}$ such that
\label{RHP m tilde}
    \begin{enumerate}[label={(\arabic*)}]
        \item $\tilde{m}(z)$ is analytic off $\Sigma_{j,1}$;
        \item $\tilde{m}_+=\tilde{m}_-v^{[2]},\quad z\in \Sigma_{j,1};$
        \item $\tilde{m}=I+\OO(z^{-1}),\quad z\rightarrow\infty$.
    \end{enumerate}
    \end{problem}
    Since $E_{j,1}|_{\Sigma_{j,1}}$ is analytic near $\Sigma_{j,1}$ for $z$ away from $z_j$, and enjoys a factorization\footnote{$(w^-,w^+)$ will be called the factorization data for the jump matrix.}:
    \begin{align*}
        (I-w^-)^{-1}(I+w^+),
    \end{align*}
    where 
\begin{align*}
    w^-&=I-(v^{[2]})^{-1}=(v^{[2]})-I,\\
    w^+&=0,
\end{align*}
and the superscribes $\pm$ indicate the analyticitiy in the left/right neighborhood of the the contour.
    
    Following the definition in \cite{BC84}, we define the Beals-Coifman operator, for any $f\in L^2(\Sigma_{j,1})$, as follows:
    \begin{align*}
        C_w(f)=C_+(fw^-)+C_-(fw^+),
    \end{align*}
    where $C$ means the usual Cauchy operator, i.e., 
    $$Cf(z)=\frac{1}{2\pi i}\int_{\Sigma_{j,1}}\frac{f(s)}{s-z}\text{ds},$$
    and $C_{\pm}$ means the non-tangential limits from left/right side.

    The following proposition, which plays a fundamental role in Deift-Zhou's method, is well-known.
    \begin{prop}[see also proposition 2.11 in \cite{DZ03}]
    If $\mu\in I+L^2$ solves the singular integral equation:
    \begin{align}
    \label{sigular equation}
        \mu=I+C_w(\mu).
    \end{align}
    Then the (unique) solution to the RHP for $\tilde{m}$ reads\footnote{Here $w=w^++w^-$.}
    \begin{align}
    \label{beals coifman solution}
        \tilde{m}=I+C(\mu w).
    \end{align}
    \end{prop}

  Then follow the localization principle in \cite{DZ93,do2009nonlinear,varzugin96}, and the simple argument on the vertical segments, we arrive at a new RHP on the new contours: fixing $\rho_0>0$ small, define
  \begin{align*}
      \Sigma_{j,1}^o&:=\{z:z=z_j+\rho e^{i\alpha},\rho\in (0,\rho_0)\},\\
      \Sigma_{j,2}^o&:=\{z:z=z_j+\rho e^{i(\pi-\alpha)},\rho\in (0,\rho_0)\},\\
      \Sigma_{j,3}^o&:=\{z:z=z_j-\rho e^{i\alpha},\rho\in (0,\rho_0)\},\\
      \Sigma_{j,4}^o&:=\{z:z=z_j+\rho e^{i(\alpha-\pi)},\rho\in (0,\rho_0)\}.
  \end{align*}
  Then with the new contour (see Fig.\ref{truncated contours}) $\Sigma^o=\cup_{j=1,..,l,k=1,2,3,4}\Sigma_{j,k}^o$, the new RHP reads as follows:
  \begin{problem}
  Looking for a 2 by 2 matrix-valued function $\hat{m}^{[2]}$ such that
\label{RHP m hat}
  \begin{enumerate}[label={(\arabic*)}]
      \item $\hat{m}^{[2]}_+=\hat{m}^{[2]}_-\hat{v}^{[2]},\quad z\in \Sigma^o,$ with
  $
          \hat{v}^{[2]}=\tilde{v}^{[2]}|_{\Sigma^0};
$
      \item $\hat{m}^{[2]}=I+\OO(z^{-1}),\quad z\rightarrow\infty.$
  \end{enumerate}
  \end{problem}

\begin{figure}[!h]
\centering
\begin{tikzpicture}
\draw [dashed,color=gray,opacity=0.3,line width=0.35 mm] (-6,0)--(6,0);
\draw (-6,1/2)--(-2,-1.5);
\draw (-6,-1/2)--(-2,1.5);
\draw [dashed,color=gray,opacity=0.3,line width=0.35 mm](-2,-1.5)--(0,-2.5);
\draw [dashed,color=gray,opacity=0.3,line width=0.35 mm](-2,1.5)--(0,2.5);
\draw (6,-1/2)--(2,1.5);
\draw (6,1/2)--(2,-1.5);
\draw [dashed,color=gray,opacity=0.3,line width=0.35 mm](2,1.5)--(0,2.5);
\draw [dashed,color=gray,opacity=0.3,line width=0.35 mm](2,-1.5)--(0,-2.5);
\draw [dashed,color=gray,opacity=0.3,line width=0.35 mm](0,2.5)--(0,-2.5);
\node [below] at (-5,-0.2) {$z_j$} ;
\node [below] at (5,-0.2) {$z_{j+1}$};
\node [above] at (-3.4,0.8) {$\Sigma_{j,1}^o$};
\node [above] at (3.4,0.8) {$\Sigma_{j+1,2}^o$};
\node [below] at (-3.4,-0.8) {$\Sigma_{j,4}^o$};
\node [below] at (3.4,-0.9) {$\Sigma_{j+1,3}^o$};
\end{tikzpicture}
\caption{New contours, dashed line segments are those deleted parts.}
\label{truncated contours}
\end{figure}
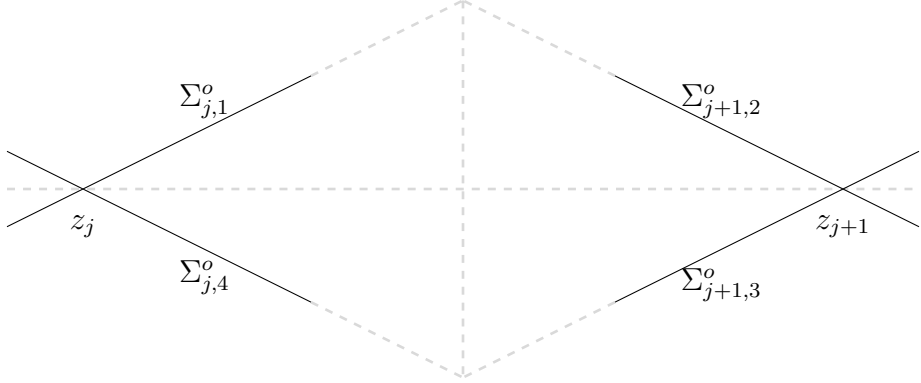

Moreover, since the potential of the mKdV hierarchy can be recovered by the formula \eqref{potential recover formula}, which can also be written as the Beals-Coifman solution:
\begin{align}
    q_{RHP}(x,t)=-\frac{1}{2\pi i}\int_{\Sigma}((I-C_w)^{-1}I)w(s)\text{ds}.
\end{align}
Then, by localization, we have
\begin{align}
    \int_{\Sigma}((I-C_w)^{-1}I)w(s)\text{ds}=\int_{\Sigma^o}((I-C_w)^{-1}I)w(s)\text{ds}+\OO(t^{-1}), \quad t\rightarrow\infty,
\end{align}
where $\Sigma$ is the the contour before localization and $w$ can be easily defined in each cross since the jumps are all triangle matrices and all entries in the diagonal are one. Let us denote 
\begin{align}
    \label{potential from RHP}
    q^o_{RHP}(x,t)=-\frac{1}{2\pi i}\int_{\Sigma^o}((I-C_w)^{-1}I)w(s)\text{ds}.
\end{align}
Then from the localization principal, we have
\begin{align}
\label{q rhp to q o rhp}
    q_{RHP}(x,t)=q^o_{RHP}(x,t)+\OO(t^{-1}),\quad t\rightarrow\infty.
\end{align}
Moreover, we define the RHP ($m^{[3]}$) which corresponds to the local Beals-Coifman solution (i.e. $q^o_{RHP}$) as follows:
\begin{problem}
\label{RHP m 3}
Looking for a 2 by 2 matrix-valued function $m^{[3]}$ such that
\begin{enumerate}[label={(\arabic*)}]
\item $m^{[3]}_+=m^{[3]}_-v^{[3]}(z),\quad z\in \Sigma^o$, with jump matrix reads
$
    v^{[3]}=\hat{v}^{[2]}\upharpoonright_{\Sigma^o};
$
 \item $m^{[3]}=I+\OO{(z^{-1})},\quad z\rightarrow\infty$.
\end{enumerate}
\end{problem}
However, the integral $\int_{\Sigma^o}((I-C_w)^{-1}I)w(s)\text{ds}$ is still hard to compute, and 
following the Deift-Zhou method, we need to separate the contributions from each stationary phase point. Thus, we need the following important lemma. 

\begin{lemma}[see equation (3.64) or proposition 3.66 in \cite{DZ93}]
\label{seperation lemma}
As $t\rightarrow\infty$,
\begin{align}
\int_{\Sigma^o}(\left(1-C_w\right)^{-1}I)w=\sum_{j=1}^{l}\int_{\Sigma^o_j}((1-C_{w_j})^{-1}I)w_j+\mathcal{O}(t^{-1}),
\end{align}
where $w_j$ is the factorization data supported on $\Sigma^o_j=\cup_{k=1}^4\Sigma^o_{j,k}$, $w=\sum_{j=1}^lw_j$ and $\Sigma^o=\cup_{j}^l\Sigma^o_j$.
\end{lemma}
\begin{proof}
First, recall the following observation by Varzugin \cite{varzugin96},
\begin{align*}
(1-C_w)(1+\sum_jC_{w_j}(1-C_{w_j})^{-1})=1-\sum_{j\neq k}C_{w_j}C_{w_k}(1-C_{w_k})^{-1}.
\end{align*}
With the hints from this observation, we need to estimate the norms of $C_{w_j}C_{w_k}$ from $L^\infty$ to $L^2$ and from $L^2$ to $L^2$. Also from next section (with a small norm argument),
we know $(1-C_{w_j})^{-1}$ are uniformly bounded in $L^2$ sense. 
Now let us focus on the contour $\Sigma^o_{j,1}$, and $\varepsilon=1$. Then the nontrivial 
entry of the factorization data is $E_{j,1}(z)e^{-2it\theta(z)},z\in \Sigma^o_{j,1}$, and thus we have
\begin{align*}
|w_{j}\upharpoonright_{\Sigma^o_{j,1}}|\leq ce^{-2t\tan(\alpha)u^2},
\end{align*} 
which implies that $\|w_{j}\upharpoonright_{\Sigma^o_{j,1}}\|_{L^1}=\mathcal{O}(t^{-1/2})$ and $\|w_{j}\upharpoonright_{\Sigma^o_{j,1}}\|_{L^2}=\mathcal{O}(t^{-1/4})$. Then following exactly the same steps in the proof of \cite{DZ93}, Lemma 3.5, we have for $j\neq k$
\begin{align*}
\|C_{w_j}C_{w_k}\|_{L^2(\Sigma^o)}&=\OO(t^{-1/2}),\\
\|C_{w_j}C_{w_k}\|_{L^\infty\rightarrow L^2(\Sigma^o)}&=\OO(t^{-3/4}).
\end{align*}
Then use the resolvent identities and the Cauchy-Schwartz inequality, 
\begin{align*}
((1-C_w)^{-1}I)&=I+\sum_{j=1}^{l}C_{w_j}(1-C_{w_j})^{-1}I\\
&+[1+\sum_{j=1}^{l}C_{w_j}(1-C_{w_j})^{-1}][1-\sum_{j\neq k}C_{w_j}C_{w_k}(1-C_{w_k})^{-1}]^{-1}\\
&(\sum_{j\neq k}C_{w_j}C_{w_k}(1-C_{w_k})^{-1})I\\
&=I+\sum_{j=1}^{l}C_{w_j}(1-C_{w_j})^{-1}I+ABDI,
\end{align*}
where
\begin{align*}
    A&:=1+\sum_{j=1}^{l}C_{w_j}(1-C_{w_j})^{-1},\\
    B&:=[1-\sum_{j\neq k}C_{w_j}C_{w_k}(1-C_{w_k})^{-1}]^{-1},\\
    D&:=\sum_{j\neq k}C_{w_j}C_{w_k}(1-C_{w_k})^{-1},
\end{align*}
and thus
\begin{align*}
|\int_{\Sigma^o}ABDIw|&\leq \|A\|_{L^2}\|B\|_{L^2}\|D\|_{L^\infty\rightarrow L^2}\|w\|_{L^2}\\
&\leq c t^{-3/4}t^{-1/4}=\OO{(t^{-1})}.
\end{align*}

Then applying the restriction lemma (\cite{DZ93}, Lemma 2.56), we have
\begin{align*}
\int_{\Sigma^o}(I+C_{w_j}(1-C_{w_j})^{-1}I)w\upharpoonright_{\Sigma^o_j}
&=\int_{\Sigma^o_j}(I+C_{w_j}(1-C_{w_j})^{-1}I)w\\
&=\int_{\Sigma^o_j}((1-C_{w_j})^{-1}I)w_j.
\end{align*}
Therefore, the proof is done.
\end{proof}



\section{A model Riemann-Hilbert problem}
\label{section: model RHP}
In the previous section, we have reduce the global RHP to $l$ local RHPs near each stationary phase point due to Lemma \ref{seperation lemma}. In fact, near each stationary phase point, we need to compute the integral $\int_{\Sigma_j^o}((1-C_{w_j})^{-1}I)w_j$, which is equivalent to a local RHP. In this section, we will approximate the local RHPs by a model RHP which can be solved explicitly by solving a parabolic-cylinder equation. Consider the following RHP:
\begin{problem}
\label{model RHP}
Looking for a 2 by 2 matrix-valued function $P(\xi;R)$ such that
\begin{enumerate}[label={(\arabic*).}]
\item $P_+(\xi;R)=P_-(\xi;R)J(\xi),\xi\in \R$, where
\begin{align*}
J(\xi)=\begin{pmatrix}
1-|R|^2&-\bar{R}\\
R&1
\end{pmatrix}
\end{align*}
is a constant matrix with respect to $\xi$ and the constant $R$ satisfies $|R|<1$;
\item $P(\xi;R)=\xi^{i\eta \sigma_3}e^{-i\frac{\xi^2}{4}\sigma_3}(I+
P_1\xi^{-1}+
\OO(\xi^{-2})),\quad
\xi\rightarrow\infty$, where $P_1=\begin{pmatrix}
0&\beta\\
\bar{\beta}&0
\end{pmatrix}$.
\end{enumerate}
\end{problem}
Then by Liouville's argument, $P'P^{-1}$ is analytic and thus
\begin{align}
P'(\xi)=(-\frac{i\xi}{2}\sigma_3-\frac{i}{2}[\sigma_3,P_1])P(\xi),
\end{align}
which can be solved in terms of the parabolic-cylinder equation, and apply the asymptotics formulas we can eventually determine that
\begin{align}
\beta=\frac{\sqrt{2\pi}e^{i\pi/4}e^{-\pi\eta/2}}{R\Gamma(-a)},
\end{align}
where 
\begin{align}
a=i\eta,
\end{align}
with $\eta=-\frac{1}{2\pi}\log{(1-|R|^2)}$.

The above result has been presented in the literature\footnote{The first description of this model RHP was presented by A. R. Its \cite{its1981asymptotics}. Later examples of the model can be find in \cite{DZ93,DZ03,dieng2008longtime,do2009nonlinear,varzugin96, Ma2019_3mkdv, ma2020long}.  } in many ways. Here we follows the representations in \cite{DZ93}. Next, we will connect this model RHP to the original RHP. Recall, near stationary phase point $z_j$, we need to estimate integral $\int_{\Sigma_j^o}((1-C_{w_j})^{-1}I)(w_{j+}+w_{j-})$, which is equivalent to solve the following RHP ($m^{[3,j]},\quad j=1,\cdots ,l$):
\begin{enumerate}[label={(\arabic*)}]
\item $m^{[3,j]}_+(z)=m^{[3,j]}_-(z)v^{[3,j]}(z),z\in \Sigma_j^o$. The jump matrix reads
\begin{align}
v^{[3,j]}(z)=\begin{cases}
\TL{R_j^\sharp(z-z_j)^{-2i\eta(z_j)}e^{-it\theta''(z_j)(z-z_j)^2}},z\in\Sigma_{j,1}^o,\\
\TU{-\frac{\bar{R}_j^\sharp}{1-|R^\sharp_j|^2}(z-z_j)^{2i\eta(z_j)}e^{it\theta''(z_j)(z-z_j)^2}},z\in \Sigma_{j,2}^o,\\
\TL{\frac{R_j^\sharp}{1-|R^\sharp_j|^2}(z-z_j)^{2i\eta(z_j)}e^{-it\theta''(z_j)(z-z_j)^2}},z\in \Sigma_{j,3}^o,\\
\TU{-\bar{R}_j^\sharp(z-z_j)^{-2i\eta(z_j)}e^{it\theta''(z_j)(z-z_j)^2}},z\in\Sigma_{j,4}^o,
\end{cases}
\end{align}
where $R_j^\sharp=R_j\delta^{-2}_je^{-2it\theta(z_j)}$;
\item $m^{[3,j]}=I+\OO(z^{-1}),\quad z\rightarrow \infty.$
\end{enumerate}

Set $\xi=(2t\theta''(z_j))^{1/2}(z-z_j)$ and by closing lenses, we arrive at an equivalent RHP on the real line:
\begin{enumerate}[label={(\arabic*)}]
\item $m^{[4,j]}(\xi)_+=m^{[4]}_-v^{[4,j]}(\xi),\xi\in \Sigma_{j}^p$.  The new jump is
\begin{align}
v^{[4,j]}(\xi)=(2\theta''(z_j)t)^{-\frac{i\eta(z_j)}{2}\text{ ad }\sigma_3}\xi^{i\eta(z_j)\text{ ad }\sigma_3}e^{-\frac{i\xi^2}{4}\text{ ad }\sigma_3}\begin{pmatrix}
1-|R_j^\sharp|^2&-\bar{R}_j^\sharp\\
R_j^\sharp&1
\end{pmatrix};
\end{align}
\item $m^{[4,j]}=I+\OO(\xi^{-1}),\xi\rightarrow\infty.$
\end{enumerate}
Comparing with the model RHP, we observe that $m^{[4,j]}(\xi)(2\theta''(z_j)t)^{-\frac{i\eta(z_j)}{2}\sigma_3}\xi^{i\eta(z_j)\sigma_3}e^{-\frac{i\xi^2}{4}\sigma_3}$ solves the model RHP, which leads to
\begin{align}
m^{[4]}_{1,12}&=\frac{\sqrt{2\pi}e^{i\pi/4}e^{-\pi\eta(z_j)/2}}{R^\sharp_j\Gamma(-i\eta(z_j))},\\
m^{[4]}_{1,21}&=\frac{-\sqrt{2\pi}e^{-i\pi/4}e^{-\pi\eta(z_j)/2}}{\bar{R}^\sharp_j\Gamma(i\eta(z_j))}.
\end{align}
Changing the variable $\xi$ back to $z$, we have
\begin{align*}
m^{[3,j]}_{1,12}(t)&= (2t\theta''(z_j))^{-\frac{1}{2}-\frac{i\eta(z_j)}{2}}\frac{\sqrt{2\pi}e^{i\pi/4}e^{-\pi\eta(z_j)/2}}{R_j^\sharp \Gamma(-i\eta(z_j))},\\
m^{[3,j]}_{1,21}(t)&=- (2t\theta''(z_j))^{-\frac{1}{2}+\frac{i\eta(z_j)}{2}}\frac{\sqrt{2\pi}e^{-i\pi/4}e^{-\pi\eta(z_j)/2}}{\bar{R}_j^\sharp \Gamma(i\eta(z_j))}.
\end{align*}
Noting that $R^\sharp_j=R_j\delta^{-2}_je^{2it\theta(z_j)}$, one can rewrite in a neat way:
\begin{align*}
m^{[3,j]}_{1,12}(t)=\frac{|\eta(z_j)|^{1/2}}{\sqrt{2t\theta''(z_j)}}e^{i\varphi(t)},\\
m^{[3,j]}_{1,21}(t)=\frac{|\eta(z_j)|^{1/2}}{\sqrt{2t\theta''(z_j)}}e^{-i\varphi(t)},
\end{align*}
where the phase is
\begin{align*}
\varphi(t)=\frac{\pi}{4}-\arg\Gamma(-i\eta(z_j))-2t\theta(z_j)-\frac{\eta(z_j)}{2}\log|2t\theta''(z_j)|+2\arg(\delta_j)+\arg(R_j).
\end{align*}
Here we have used the fact that $|\beta|^2=\eta$. Denoting
\begin{align}
\label{q as}
q_{as}(x,t)=-2i\sum_{j=1}^{l}\frac{|\eta(z_j)|^{1/2}}{\sqrt{2t\theta''(z_j)}}e^{i\varphi(t)},
\end{align}
then the connection formula \eqref{potential from RHP} and Lemma \ref{seperation lemma} lead to 
\begin{equation}
\label{q o rhp to q as}
    q^o_{RHP}(x,t)=q_{as}(x,t)+\OO(t^{-1}),\quad t\rightarrow \infty.
\end{equation}

\section{Errors from the pure $\dbar$-problem}
\label{section: dbar error}
In this section, we will discuss the error generated from the pure $\dbar$-problem of $m^{[2]}$. Let us denote
\begin{align}
\label{dbar conjugation}
E(z)=m^{[2]}({m}^{[2]}_{RHP})^{-1},
\end{align}
where $m^{[2]}_{RHP}$ denotes the solution to the pure RHP part of $m^{[2]}$.
Assuming the existence (which we will be provided in the next section), and by the normalization condition, we have
\begin{align}
E(z)=1+(m_1^{[2]}-{m}^{[2]}_{RHP,1})z^{-1}+\OO(z^{-2}),\quad z\rightarrow\infty.
\end{align}
Due to the procedure of localization and separation of the contributions, we can approximate $m^{[2]}_{RHP}$ by $\hat{m}^{[2]}$, and the error of approximating the potential is of $\OO(t^{-1})$ as $t\rightarrow\infty$.
Thus, by the equation\eqref{potential recover formula},
\begin{align}
\label{q to q rhp}
    q(x,t)=q_{RHP}(x,t)+\OO(t^{-1})+\lim_{z\rightarrow\infty}z(E-I),\quad t\rightarrow\infty.
\end{align}
Moreover, from this construction (equation \eqref{dbar conjugation}), there is no jump on the contours $\Sigma_{j,k},k=1,2,3,4,$ but only a pure $\dbar$-problem is left due to the non-analyticity. The $\dbar$-problem reads
\begin{align}
\dbar E=EW,
\end{align}
where 
\begin{align}
W(z)=m^{[2]}_{RHP}\dbar O(z)(m^{[2]}_{RHP})^{-1}.
\end{align}
From the normalization condition of $m^{[2]}_{RHP}$, we see it is uniformly bounded by $\frac{c}{1-\sup{R}}$. And to estimate the errors of recovering the potential, one actually needs to estimate $\lim_{z\rightarrow\infty}z(E-I)$, where the limit can be chosen along any rays that are not parallel to $\R$. For simplicity, we will take the imaginary axis. The $\dbar$-problem is equivalent to the following Fredholm integral equation by a simple application of the generalized Cauchy integral formula:
\begin{align}
E(z)=I-\frac{1}{\pi}\int_{\C}\frac{E(s)W(s)}{s-z}dA(s).
\end{align}
In the following, we will show for each fixed $z\in \C$, $\K_W(E)(z):=\int_{\C}\frac{E(s)W(s)}{s-z}dA(s)$ is bounded and then by the dominated convergence theorem, we will show $\lim_{z\rightarrow\infty}z(E-I)=\OO(t^{-3/4})$. First of all, since $m^{[3]}$ is uniformly bounded, upon settng $z=z_j+u+iv$, we have
\begin{align}
\|W\|_{\infty}\lesssim \begin{cases}
|\dbar E_{j,k}|e^{-2t\theta''(z_j)uv},z\in \Omega_{j,k},k=1,4,\\
|\dbar E_{j,k}|e^{2t\theta''(z_j)uv},z\in \Omega_{j,k},k=3,6,\\
\end{cases},
\end{align}
where $0\leq a\lesssim b$ means there exists $C>0$ such that $a\leq Cb$. Then we have
\begin{align}
\K_W(E)\leq \|E\|_\infty\int_\C\frac{\|W(s)\|_\infty}{|s-z|}dA(s).
\end{align}
We claim the following lemma:
\begin{lemma}\label{estimates of W}
	Let $\Omega=\{s:s=\rho e^{i\phi}, \rho\geq 0,\phi\in [0,\pi/4]\}$, and $z\in \Omega$. Then
	\begin{align}
	\int_\Omega\frac{|u^2+v^2|^{-1/4}e^{-tuv}}{|u+iv-z|}dudv=\OO(t^{-1/4}).
	\end{align}
\end{lemma}
\begin{proof}
Since there are two singularities of the integrand at $z$ and $(0,0)$. In the first case, set $z\neq 0$, and let $d=dist(z,0)$. We split $\Omega$ into three parts: $\Omega_1\cup \Omega_2\cup \Omega_3$, where $\Omega_1=\{s:|s|<d/3\}\cap \Omega$ , $\Omega_2=\{s:|s-z|<d/3\}\cap \Omega$ and $\Omega_3=\Omega\backslash (\Omega_1\cup \Omega_2)$. In the region $\Omega_1$, $|s-z|\geq 2d/3$, and thus
\begin{equation}
\begin{split}
|\int_{\Omega_1}\frac{|u^2+v^2|^{-1/4}e^{-tuv}}{|u+iv-z|}dudv|&\leq \frac{3}{2d}\int_{0}^\infty\int_0^u\frac{e^{-tuv}}{(u^2+v^2)^{1/4}}dvdu\\
& \text{ substituted $v=wu$}\\
&\leq \frac{3}{2d}\int_0^\infty\int_0^1\frac{e^{-tu^2w}}{(1+w^2)^{1/4}}u^{1/2}dwdu\\
&\leq \frac{3}{2d}\int_0^\infty\int_0^1 e^{-tu^2w}u^{1/2}dwdu\\
&=\frac{3}{2d}\int_0^\infty\frac{1-e^{-tu^2}}{tu^{3/2}}du\\
&=\frac{3}{2d}\frac{1}{2}t^{-3/4}\int_0^\infty\frac{1-e^{-u}}{u^{5/4}}du\\
&=\frac{3}{d}\Gamma(3/4)t^{-3/4}.
\end{split}
\label{ref to omega 1}
\end{equation}
In the region $\Omega_2$, $|s|^{-1/2}\leq (2d/3)^{-1/2}$, we have
\begin{align*}
|\int_{\Omega_1}\frac{|u^2+v^2|^{-1/4}e^{-tuv}}{|u+iv-z|}dudv|&\leq \sqrt{\frac{3}{2d}}\int_{\Omega_2}\frac{e^{-tuv}}{((u-x)^2+(v-y)^2)^{1/2}}dvdu\\
&\leq \sqrt{\frac{3}{2d}}\int_{0}^{d/3}\int_0^{2\pi}e^{-t(x+\rho\cos(\theta))(y+\rho\sin(\theta))}d\theta d\rho\\
&\leq \frac{2\pi}{3}\sqrt{\frac{3d}{2}}e^{-txy}.
\end{align*}
While in the region $\Omega_3$, 
\begin{align*}
|\int_{\Omega_3}\frac{|u^2+v^2|^{-1/4}e^{-tuv}}{|u+iv-z|}dudv|&\leq \int_{0}^{\infty}\int_0^u e^{-tuv}dvdu=\OO(t^{-1}).
\end{align*}
Now consider $z=0$. We have
\begin{align*}
|\int_{\Omega}\frac{e^{-tuv}}{(u^2+v^2)^{3/4}}dA(u,v)|&=\int_{0}^\infty\int_0^u\frac{e^{-tuv}}{(u^2+v^2)^{3/4}}dvdu\\
&=\int_0^\infty\int_0^1\frac{e^{-tu^2w}}{(1+w^2)^{3/4}u^{1/2}}dwdu\\
&\leq \int_0^\infty\int_0^1\frac{e^{-tu^2w}}{u^{1/2}}dwdu\\
&=\int_0^\infty\frac{1-e^{-tu^2}}{tu^{5/2}}du\\
&=\int_0^\infty\frac{1-e^{-u}}{tt^{-5/4}u^{5/4}}t^{-1/2}\frac{1}{2}u^{-\frac{1}{2}}du\\
&=\frac{1}{2}t^{-1/4}\int_0^\infty\frac{1-e^{-u}}{u^{7/4}}du\\
&=\frac{3}{8}t^{-1/4}\Gamma(1/4).
\end{align*}
By assembling all together, the proof is done.
\end{proof}
\begin{remark}
The essential fact that makes the above true is the rapid decay of the exponential factor in the region. And the lemma also tells us that those mild singularities, which have rational order growth, can be absorbed by the exponential factor. Back to our situation, after some elementary transformations (translation and rotation), the estimation of $\int_\C\frac{\|W(s)\|_\infty}{|s-z|}dA(s)$ will eventually reduce to a similar situation discussed in the above lemma.
\end{remark}
Based on Lemma \ref{estimates of W}, we know that when $t$ is sufficiently large, $\|\K_W\|<1$ and thus the resolvent is uniformly bounded, and we obtain the following estimate by taking a standard Neumann series, for some sufficiently large $t_0$,
\begin{align}
\|E-I\|_\infty=\|\K_W(1-\K_W)^{-1} I\|_\infty\leq \frac{ct^{-1/4}}{1-ct^{-1/4}}\leq ct^{-1/4}, \quad t>t_0.
\end{align} 
Now since for each $z\in \Omega_{j,k}$, we have $|\dbar E_{j,k}(z)|\leq c(|z-z_j|^{-1/2}+|R'(u+z_j)|),$ and apply the dominated convergence theorem, we have
\begin{align*}
\lim_{z\rightarrow\infty}|z(E-I)|&\leq \frac{1}{\pi}\sum_{j=1}^{l}\sum_{k=1}^4\|E\|_{L^\infty}\int_{\Omega_{j,k}}\|W\|_\infty ds,
\end{align*}
and use the Lemma \ref{estimates of W} again, we will eventually have:
\begin{align}
\label{dbar error}
E_1=\lim_{z\rightarrow\infty}|z(E-I)|=\OO(t^{-3/4}).
\end{align}

\section{Asymptotics representation}
\label{section: asymptotics representations}
First, we summarize all the steps as following (see Fig.\ref{fig:steps}):
\begin{enumerate}[label={(\arabic*)}]
\item Initial RHP $m^{[0]}=m$, see RHP \ref{RHP m 0}.
\item Conjugate initial RHP to obtain $m^{[1]}=m^{[0]}\delta^{\sigma_3}$, see RHP \ref{RHP m 1}.
\item Open lenses to obtain a mixed $\dbar$-RHP \ref{dbar RHP}.
\item Approximate the RHP part $m^{[2]}_{RHP}$ of $m^{[2]}$ by removing $\Sigma_{j+\frac{1}{2}}$ (see RHP \ref{RHP m tilde}), localization (see RHP \ref{RHP m hat}), reducing the phase function and separating the contributions (see RHP \ref{RHP m 3}). The error term is $\OO(t^{-1})$. Note those exponential decaying errors are absorbed by $\OO(t^{-1})$.
\item Comparing $m^{[2]}$ and $m^{[2]}_{RHP}$ and computing the error by analysis a pure $\dbar$-problem. The error term is $\OO(t^{-3/4})$.
\end{enumerate}

\begin{figure}[h]
    \centering
\tikzset{every picture/.style={line width=0.75pt}} 

\begin{tikzpicture}[x=0.75pt,y=0.75pt,yscale=-1,xscale=1]

\draw    (351,68.5) -- (351,95.5) ;
\draw [shift={(351,97.5)}, rotate = 270] [color={rgb, 255:red, 0; green, 0; blue, 0 }  ][line width=0.75]    (10.93,-3.29) .. controls (6.95,-1.4) and (3.31,-0.3) .. (0,0) .. controls (3.31,0.3) and (6.95,1.4) .. (10.93,3.29)   ;
\draw    (352,125.5) -- (352,159.5) ;
\draw [shift={(352,161.5)}, rotate = 270] [color={rgb, 255:red, 0; green, 0; blue, 0 }  ][line width=0.75]    (10.93,-3.29) .. controls (6.95,-1.4) and (3.31,-0.3) .. (0,0) .. controls (3.31,0.3) and (6.95,1.4) .. (10.93,3.29)   ;
\draw    (351,187.5) -- (351,219.5) ;
\draw [shift={(351,221.5)}, rotate = 270] [color={rgb, 255:red, 0; green, 0; blue, 0 }  ][line width=0.75]    (10.93,-3.29) .. controls (6.95,-1.4) and (3.31,-0.3) .. (0,0) .. controls (3.31,0.3) and (6.95,1.4) .. (10.93,3.29)   ;
\draw    (259,221.5) -- (485,222.5) ;
\draw    (259,221.5) -- (259.95,259.5) ;
\draw [shift={(260,261.5)}, rotate = 268.57] [color={rgb, 255:red, 0; green, 0; blue, 0 }  ][line width=0.75]    (10.93,-3.29) .. controls (6.95,-1.4) and (3.31,-0.3) .. (0,0) .. controls (3.31,0.3) and (6.95,1.4) .. (10.93,3.29)   ;
\draw    (485,222.5) -- (485,259.5) ;
\draw [shift={(485,261.5)}, rotate = 270] [color={rgb, 255:red, 0; green, 0; blue, 0 }  ][line width=0.75]    (10.93,-3.29) .. controls (6.95,-1.4) and (3.31,-0.3) .. (0,0) .. controls (3.31,0.3) and (6.95,1.4) .. (10.93,3.29)   ;
\draw    (485,292.5) -- (485.94,323.5) ;
\draw [shift={(486,325.5)}, rotate = 268.26] [color={rgb, 255:red, 0; green, 0; blue, 0 }  ][line width=0.75]    (10.93,-3.29) .. controls (6.95,-1.4) and (3.31,-0.3) .. (0,0) .. controls (3.31,0.3) and (6.95,1.4) .. (10.93,3.29)   ;
\draw    (487,354.5) -- (487.93,380.5) ;
\draw [shift={(488,382.5)}, rotate = 267.95] [color={rgb, 255:red, 0; green, 0; blue, 0 }  ][line width=0.75]    (10.93,-3.29) .. controls (6.95,-1.4) and (3.31,-0.3) .. (0,0) .. controls (3.31,0.3) and (6.95,1.4) .. (10.93,3.29)   ;
\draw    (488,404.5) -- (489.86,430.51) ;
\draw [shift={(490,432.5)}, rotate = 265.90999999999997] [color={rgb, 255:red, 0; green, 0; blue, 0 }  ][line width=0.75]    (10.93,-3.29) .. controls (6.95,-1.4) and (3.31,-0.3) .. (0,0) .. controls (3.31,0.3) and (6.95,1.4) .. (10.93,3.29)   ;
\draw    (472,449.5) .. controls (434.38,430.69) and (418.32,313.87) .. (456.82,282.43) ;
\draw [shift={(458,281.5)}, rotate = 503.13] [color={rgb, 255:red, 0; green, 0; blue, 0 }  ][line width=0.75]    (10.93,-3.29) .. controls (6.95,-1.4) and (3.31,-0.3) .. (0,0) .. controls (3.31,0.3) and (6.95,1.4) .. (10.93,3.29)   ;
\draw    (263,289.5) .. controls (302,323.5) and (420,305.5) .. (460,275.5) ;
\draw    (360,306.5) -- (277.91,331.91) ;
\draw [shift={(276,332.5)}, rotate = 342.8] [color={rgb, 255:red, 0; green, 0; blue, 0 }  ][line width=0.75]    (10.93,-3.29) .. controls (6.95,-1.4) and (3.31,-0.3) .. (0,0) .. controls (3.31,0.3) and (6.95,1.4) .. (10.93,3.29)   ;
\draw    (595,234.5) .. controls (557.57,251.25) and (525.96,253.44) .. (525.02,289.81) ;
\draw [shift={(525,291.5)}, rotate = 270] [color={rgb, 255:red, 0; green, 0; blue, 0 }  ][line width=0.75]    (10.93,-3.29) .. controls (6.95,-1.4) and (3.31,-0.3) .. (0,0) .. controls (3.31,0.3) and (6.95,1.4) .. (10.93,3.29)   ;
\draw    (595,234.5) .. controls (652.13,331.03) and (569.55,326.66) .. (541.25,353.26) ;
\draw [shift={(540,354.5)}, rotate = 313.96000000000004] [color={rgb, 255:red, 0; green, 0; blue, 0 }  ][line width=0.75]    (10.93,-3.29) .. controls (6.95,-1.4) and (3.31,-0.3) .. (0,0) .. controls (3.31,0.3) and (6.95,1.4) .. (10.93,3.29)   ;
\draw    (595,234.5) .. controls (688.53,270.32) and (635.54,350.69) .. (586.74,402.72) ;
\draw [shift={(586,403.5)}, rotate = 313.3] [color={rgb, 255:red, 0; green, 0; blue, 0 }  ][line width=0.75]    (10.93,-3.29) .. controls (6.95,-1.4) and (3.31,-0.3) .. (0,0) .. controls (3.31,0.3) and (6.95,1.4) .. (10.93,3.29)   ;

\draw (350.5,50.25) node   [align=left] {\begin{minipage}[lt]{55.08pt}\setlength\topsep{0pt}
$ $$\displaystyle m^{[ 0]} =m$
\end{minipage}};
\draw (366,110) node   [align=left] {\begin{minipage}[lt]{88pt}\setlength\topsep{0pt}
$\displaystyle m^{[1]} =m^{[ 0]} \delta ^{\sigma _{3}}$
\end{minipage}};
\draw (353,71.5) node [anchor=north west][inner sep=0.75pt]   [align=left] {Conjugation};
\draw (367,180) node   [align=left] {\begin{minipage}[lt]{88pt}\setlength\topsep{0pt}
$\displaystyle m^{[ 2]} =m^{[ 1]} O( z)$
\end{minipage}};
\draw (354,128.5) node [anchor=north west][inner sep=0.75pt]   [align=left] {Open lenses};
\draw (353,190.5) node [anchor=north west][inner sep=0.75pt]   [align=left] {Mixed $\displaystyle \overline{\partial }$-RHP};
\draw (224,263) node [anchor=north west][inner sep=0.75pt]   [align=left] {$\displaystyle \overline{\partial }$-Problem};
\draw (489,279.75) node   [align=left] {\begin{minipage}[lt]{35.36pt}\setlength\topsep{0pt}
$\displaystyle m_{RHP}^{[ 2]}$
\end{minipage}};
\draw (545.5,308.75) node   [align=left] {\begin{minipage}[lt]{82.28pt}\setlength\topsep{0pt}
Remove $\displaystyle \Sigma _{j+1/2}$
\end{minipage}};
\draw (492.5,339.75) node   [align=left] {\begin{minipage}[lt]{27.88pt}\setlength\topsep{0pt}
$\displaystyle \tilde{m}^{[ 2]}$
\end{minipage}};
\draw (499,393.13) node   [align=left] {\begin{minipage}[lt]{34pt}\setlength\topsep{0pt}
$\displaystyle \hat{m}^{[ 2]}$
\end{minipage}};
\draw (531,368.25) node   [align=left] {\begin{minipage}[lt]{58.48pt}\setlength\topsep{0pt}
Localization
\end{minipage}};
\draw (500.5,446.25) node   [align=left] {\begin{minipage}[lt]{31.96pt}\setlength\topsep{0pt}
$\displaystyle m^{[ 3]}$
\end{minipage}};
\draw (496,407) node [anchor=north west][inner sep=0.75pt]   [align=left] {{\scriptsize Reduce phase \& separate contributions}};
\draw (426.5,435.75) node [font=\small]   [align=left] {\begin{minipage}[lt]{49.64pt}\setlength\topsep{0pt}
Existence
\end{minipage}};
\draw (400.5,338.25) node  [font=\small] [align=left] {\begin{minipage}[lt]{60.52pt}\setlength\topsep{0pt}
Small norm\\(for large $\displaystyle t$)
\end{minipage}};
\draw (264,363) node   [align=left] {\begin{minipage}[lt]{128pt}\setlength\topsep{0pt}
$\displaystyle E=m^{[ 2]}\left( m_{RHP}^{[ 2]}\right)^{-1}$
\end{minipage}};
\draw (279,304) node [anchor=north west][inner sep=0.75pt]   [align=left] {\textcolor[rgb]{0.94,0.05,0.05}{Error}};
\draw (583,214) node [anchor=north west][inner sep=0.75pt]   [align=left] {\textcolor[rgb]{0.94,0.05,0.05}{Error}};
\draw (570.5,223.25) node  [font=\footnotesize] [align=left] {\begin{minipage}[lt]{42.84pt}\setlength\topsep{0pt}
$\displaystyle \mathcal{O}\left( t^{-1}\right)$
\end{minipage}};
\draw (254,312.75) node  [font=\footnotesize] [align=left] {\begin{minipage}[lt]{46.24pt}\setlength\topsep{0pt}
$\displaystyle \mathcal{O}\left( t^{-3/4}\right)$
\end{minipage}};

\end{tikzpicture}

    \caption{Steps of the $\dbar$-steepest method.}
    \label{fig:steps}
\end{figure}
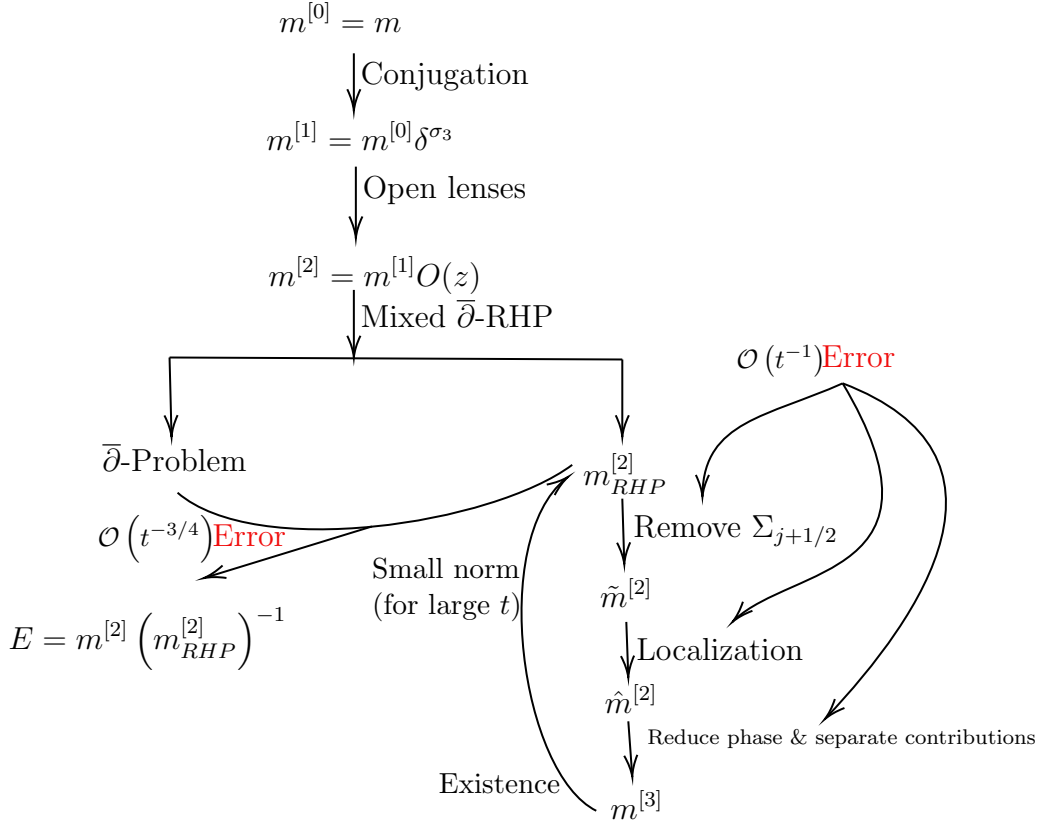

Now by undoing all the steps, we arrive at:
\begin{align*}
m^{[0]}(z)=E(z)m^{[2]}_{RHP}(z)O^{-1}(z)\delta^{-\sigma_3}.
\end{align*}
Since $O(z)$ uniformly converges to $I$ as $z\rightarrow \infty$, and $\delta^{-\sigma_3}$ is a diagonal matrix, they do not affect the recovering of the potential. Thus we obtain
\begin{align*}
q(x,t)&=-2i(m^{[2]}_{RHP,1,12}+E_{1,12})\\
&=q_{RHP}(x,t)-2iE_{1,12}\\
& \text{ by \eqref{q rhp to q o rhp},\eqref{q o rhp to q as}}\\
&=q_{as}(x,t)+\OO(t^{-1})-2iE_{1,12}\\
& \text{ by \eqref{dbar error}}\\
&=q_{as}(x,t)+\OO(t^{-1})+\OO(t^{-3/4})\\
&=q_{as}(x,t)+\OO(t^{-3/4}),
\end{align*}
where $q_{as}(x,t)$ is given by equation \eqref{q as}.
\begin{remark}
    Note that due to the analysis in the section 7, according to Proposition 2.6 and Proposition 2.11 of \cite{DZ03}, together with the small norm theory, the existence and uniqueness of the model RHP implies, via the estimates of the corresponding Beals-Coifman operators, the existence and uniqueness of RHP \ref{RHP m 3}. Similarly, we obtain the existence and uniqueness of $\tilde{m}^{[2]}$,   $\hat{m}^{[2]}$ and eventually $m^{[2]}_{RHP}$.
\end{remark}
\begin{remark}
	From equation \eqref{q as}, we know $q_{as}$ is $\OO(t^{-1/2})$ as $t\rightarrow\infty$ in the region $x<0$ and consider the limit along the ray $x=-ct$ for some positive constant $c$.
\end{remark}

\section{Fast decaying region}
\label{section: fast decaying}
In this section and the next section, we will focus only on the case of the defocusing mKdV flow. In this case, the phase function reads
\begin{align*}
    \theta(z;x,t)=\frac{x}{t}z+cz^n,\quad n \text{ is an odd positive integer.}
\end{align*}
In the previous sections, we have derived the asymptotic solutions to the defocusing mKdV flow in the oscillating region, namely, along the ray $x=-\nu t,\ \nu>0,\  t\rightarrow\infty$. In this section, we consider the long-time behavior along the ray $x=\nu t,\ \nu>0,\ t\rightarrow\infty$, which we call it the fast decaying region as we will soon prove in this region, the solution decay like $\OO(t^{-1})$, which is faster than the leading term in the oscillating region, i.e., $\OO(t^{-1/2})$, as $t\rightarrow\infty$.

In the fast decaying region, the phase function enjoys the following properties: 
\begin{enumerate}[label={(\arabic*)}]
	\item There exits $\epsilon=\epsilon(n,\nu)>0$ such that $\pm \Im(\theta)>0$ in the strips $\{z:\pm\Im(z)\in (0,\epsilon)\}$, respectively.
	\item There exits $M\in (0,1/\epsilon)$ such that $\Im(\theta) \geq nvu^{n-1}$ for $|u|\geq M\epsilon$ and $\Im(\theta)\geq v(1-(M\epsilon)^2)$ for $|u|\leq M\epsilon$. Here $z=u+iv$.
\end{enumerate}

First we will formulate the RHP as follows:
\begin{problem}

\label{RHP for fast decaying}
Given $R(z)\in H^{1,1}(\R)$, looking for a 2 by 2 matrix-value function $m$ such that 
\begin{enumerate}[label={(\arabic*)}]
	\item $m_+=m_-e^{-it\theta(z)\ad{\sigma_3}}v(z),z\in \R,$
	where the jump matrix is given by
	\begin{align}
	v(z)=\begin{pmatrix}
	1-|R|^2& -\bar{R}\\
	R & 1\\
	\end{pmatrix}=\begin{pmatrix}
	1 & -\bar{R}\\
	0& 1\\
	\end{pmatrix}\begin{pmatrix}
	1 & 0\\
	R & 1\\
	\end{pmatrix};
	\end{align}
	\item $m=I+\OO(z^{-1}),\quad z\rightarrow \infty$.
\end{enumerate}
\end{problem}
\begin{theorem}
	For the above RHP, the solution $m$ enjoys the following asymptotics as $t\rightarrow\infty$:
	\begin{align}
	m_1(t)=\OO(t^{-1}).
	\end{align}
	where  $m=I+m_1(t)/z+\OO(z^{-2}),\  z\rightarrow\infty$.
\end{theorem}

\begin{figure}[ht]
    \centering
    \includegraphics[width=0.85\textwidth]{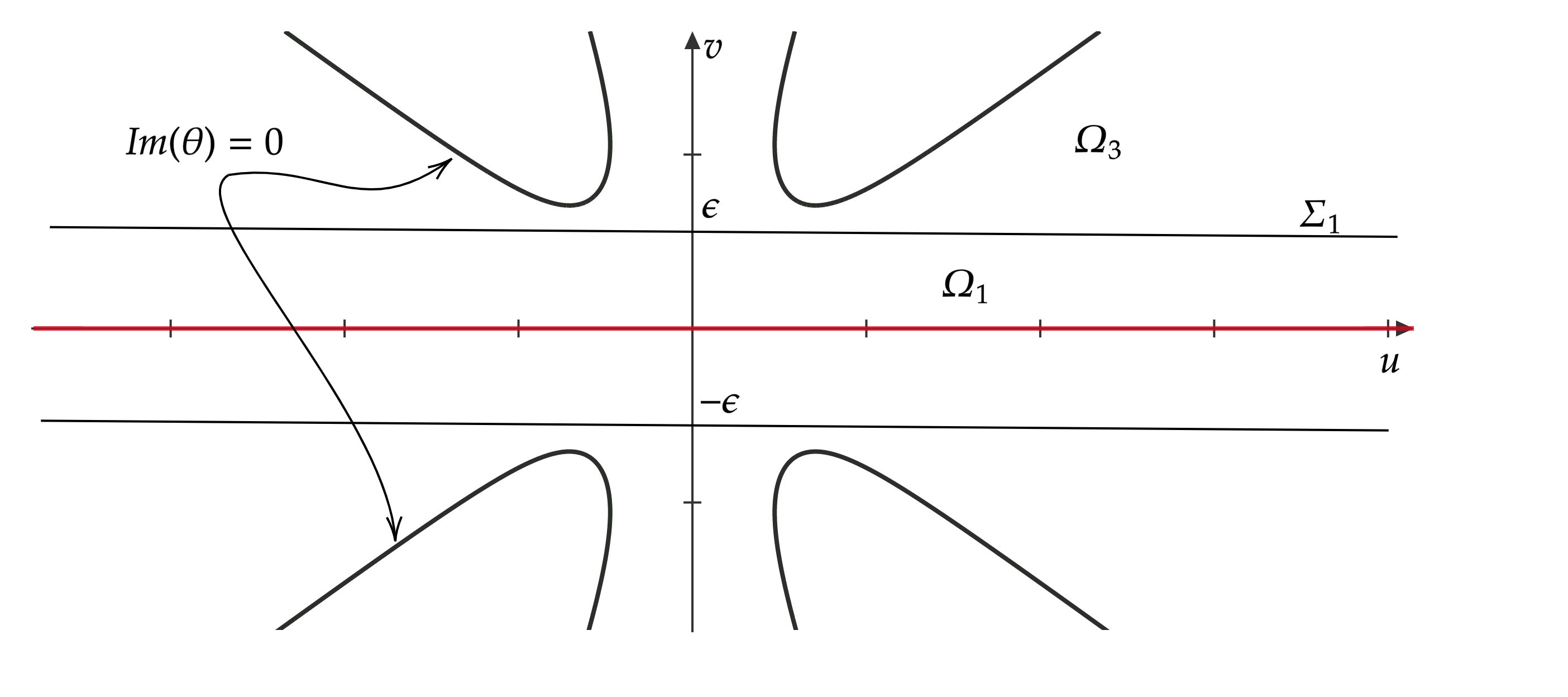}
    \caption{$\dbar$-extension for the case of the fast decaying region. Here we only draw the case when $n=5$. For generic odd $n$, there are $\frac{n-1}{2}$ curves of $\Im{\theta}=0$ in the upper and in the lower half plane. }
    \label{fig: fast decaying }
\end{figure}

\begin{proof}
	 In the light of $\dbar$-steepest descent, to open the lens, we multiple a smooth function $O(z)$ to $m$, where $O(z)$ is given by
	\begin{align*}
	O(z)=\begin{cases}
	\begin{pmatrix}
	1 & 0\\
	\frac{{-R}(\Re{z})e^{2it\theta(z)}}{1+(\Im{z})^2} & 1\\
	\end{pmatrix}, \quad z\in {\Omega}_1,\\
		\begin{pmatrix}
	1 &	\frac{\bar{R}(\Re{z})e^{-2it\theta(z)}}{1+(\Im{z})^2}\\
0 &1\\
	\end{pmatrix},\quad z\in \Omega^*_{1},\\
	I,\quad z \in \C\backslash (\Omega_1\cup\Omega^*_1),
	\end{cases}
	\end{align*} 
	where (see Fig.\ref{fig: fast decaying })
	\begin{align*}
	    \Omega_1&=\{z:\Im{z}\in(0,\epsilon) \},\\
	    \Omega^*_1&=\{z:\Im{z}\in(-\epsilon,0) \}.
	\end{align*}
	
	Let us denote $\Sigma_1=\{z:\Im{z}=\epsilon\}$, see Fig.\ref{fig: fast decaying }, and let $$\tilde{m}=mO,\quad z\in \C.$$
	
	Now as usual, we obtain a $\dbar$-RHP, due to the exponential decaying of the off-diagonal term, and the jump matrix of the RHP part will approach $I$. Hence by a small norm argument, we know the solution will close to $I$ as $z\rightarrow\infty$. Denote the solution to the pure RHP by $m^\sharp$, and small norm theory leads to $m^\sharp=I+\OO(e^{-c(\epsilon)t}),c(\epsilon)>0,z\rightarrow\infty$. Next, consider
	\begin{align}
	E=\tilde{m}(m^\sharp)^{-1}.
	\end{align}
	By direct computation one can show $E$ doesn't have any jump on $\Sigma_1$ and it satisfies a pure $\dbar$-problem:
	\begin{align}
	\dbar E=EW,
	\end{align}
	where \begin{align*}W=\begin{cases}
	\begin{pmatrix}
	0 & 	m^{\sharp}e^{-2it\theta(z)}\dbar(\frac{\bar{R}(\Re{z})}{1+(\Im{z})^2}) (m^\sharp)^{-1} \\
	0 & 0
	\end{pmatrix},\quad z\in \Omega_1,\\
	\begin{pmatrix}
	0 & 0\\
	m^{\sharp}e^{2it\theta(z)}\dbar(\frac{-{R}(\Re{z})}{1+(\Im{z})^2}) (m^\sharp)^{-1} & 0
	\end{pmatrix},\quad z\in \Omega_1^*,\\
	0,\quad z\in \C\backslash (\Omega_1\cup\Omega_1^*),
	\end{cases}
\end{align*} 
where $\dbar=\frac{1}{2}(\partial_{\Re{z}}+i\partial_{\Im{z}})$. 
	
	Since $R,\bar{R}\in H^{1,1}$, $\dbar(\frac{\bar{R}(\Re{z})}{1+(\Im{z})^2}),\dbar(\frac{-{R}(\Re{z})}{1+(\Im{z})^2})$ are uniformly bounded by some non-negative $L^2(\R)$ function $f(\Re{z})$. Note that $m^\sharp$ is uniformly close to $I$, and setting $z=u+iv$, and considering $z\in \Omega_1$ first,  we have
	\begin{align*}
	\|W\|_\infty\leq f(u)e^{-t\Im{\theta(u,v)}}, \forall u\in \R, v\in (0,\epsilon).
	\end{align*}
	By the same procure as the one in section 8, the error of approximating $m$ by the identity matrix is given by the following integral (since there is only one non-trivial entry of $W$):
	\begin{align}
	\Delta:=\int_0^\epsilon\int_\R f(u)e^{-t\Im{\theta}}dudv.
	\end{align}
	Split the $u$ into two regions: (1) $|u|\leq M\epsilon$, (2) $|u|\geq M\epsilon$. And denote them by $\Delta_1$, $\Delta_2$ respectively. Then $\Delta=\Delta_1+\Delta_2$. And 
	\begin{align*}
	\Delta_1&\leq \int_0^\epsilon\int_{-M\epsilon}^{M\epsilon}f(u)e^{-tv(1-M^2\epsilon^2)}dudv\\
	&\text{by Cauchy-Schwartz }\\
	&\leq \|f\|_{L^2(\R)}(2M\epsilon)^{1/2}\frac{1-e^{-t\epsilon(1-M^2\epsilon^2)}}{t(1-M^2\epsilon^2)}\\
	&=\OO(t^{-1}).
	\end{align*}
	On the other hand,
	\begin{align*}
	\Delta_2&\leq \int_0^\epsilon\int_{|u|\geq M\epsilon}f(u)e^{-ntvu^{n-1}}dudv\\
	&=\int_{|u|\geq M\epsilon} f(u)\int_0^\epsilon e^{-ntvu^{n-1}}dvdu\\
	&\leq t^{-1}\|f\|_{L^2}(\int_{|u|\geq M\epsilon}(\frac{1-e^{-ntvu^{n-1}}}{nu^{n-1}})^2du)^{1/2}\\
	&\leq t^{-1}\|f\|_{L^2} \frac{n}{n-2}(M\epsilon)^{-(n-2)}\\
	&=\OO(t^{-1}).
	\end{align*}
	Similarly, we can prove that for $z\in\Omega_1^*$, we also have the error estimate $\OO(t^{-1})$.
	Assembling all together, we conclude that the error term is $\OO(t^{-1})$, and  $m_1=\OO(t^{-1})$, as $ t\rightarrow\infty.$
\end{proof}
\section{Painlev\'e region}
\label{section: painleve region}
In this section, we first derive the Painlev\'e II hierarchy based on some RHP. Then, we will connect the long-time behavior of the mKdV hierarchy in the so-called Painlev\'e region to solutions of the Painlev\'e II hierarchy.

\subsection{Painlev\'e II hierarchy}
As mentioned in \cite{PhysRevLett.38.1103}, the mKdV equation is can be transferred to the Painlev\'e II equation. The authors in \cite{PhysRevLett.38.1103} also suggest the connection between integrable PDEs with Painlev\'e equations. In \cite{ClarksonMKdVLax}, the authors explicitly derived the Painlev\'e II hierarchy from self-symmetry reduction of the mKdV hierarchy (see page 59 of \cite{ClarksonMKdVLax}. And also \cite{Clarkson_1999}).  In this section, we will provide a slight different (as comparing to \cite{Clarkson_1999}) algorithm  based on Riemann-Hilbert problems to generate the Painlev\'e II hierarchy. Let's denote $\Theta(x,z)=xz+\frac{c}{n}z^n$, and suppose $Y$ solves the following RHP:
\begin{align*}
Y_+&=Y_-e^{i\Theta\sigma_3}v_0e^{-i\Theta\sigma_3},\quad z\in \Sigma_n,\\
Y&=I+\OO(z^{-1}),\quad z\rightarrow \infty.
\end{align*}
where the contour $\Sigma_n$ consists of all stokes lines $\{z:\Im{\Theta(z)}=0\}$ and $v_0$ is a constant 2 by 2 matrix that is independent of $x,z$.

Now let $\tilde{Y}=Ye^{i\Theta \sigma_3}$, and we arrive at a new RHP:
\begin{align*}
\tilde{Y}_+&=\tilde{Y}_-v_0,\quad z\in \Sigma_n,\\
\tilde{Y}&=(I+\OO(z^{-1}))e^{i\Theta \sigma_3},\quad z\rightarrow \infty.
\end{align*}
Since $v_0$ is constant, it is easily to check, by Louisville's argument, that both $\partial_z\tilde{Y}\tilde{Y}^{-1}$ and $\partial_x\tilde{Y}\tilde{Y}^{-1}$ are polynomial of $z$. Hence we obtain the following two differential equations:
\begin{align}
\partial_x\tilde{Y}\tilde{Y}^{-1}&=A(x,z),\\
\partial_z\tilde{Y}\tilde{Y}^{-1}&=B(x,z).
\end{align}
If we assume 
\begin{align}
Y&=I+\sum_{j=1}^{n-1}{Y_j(x)z^{-j}}+\OO(z^{-n}),\quad z\rightarrow\infty,\\
\underline{Y}&=Y^{-1}=I+\sum_{j=1}^{n-1}{\underline{Y}_j(x)z^{-j}}+\OO(z^{-n}),\quad z\rightarrow\infty,
\end{align}
then a direct computation shows 
\begin{align*}
A&=i[Y_1,\sigma_3]+iz\sigma_3,\\
B&=ix\sigma_3+icz^{n-1}\sigma_3+icz^{n-2}[Y_1,\sigma_3]\\
&+\sum_{k=2}^{n-1}icz^{n-1-k}(Y_k\sigma_3+\sigma_3\underline{Y}_k+\sum_{j=1}^{k-1}Y_{k-j}\sigma_3\underline{Y}_j).
\end{align*}
Since $Y_{x,z}=Y_{z,x}$, we have
\begin{equation}
\label{compitabile condition}
A_z-B_x+[A,B]=0.
\end{equation}
 Set 
\begin{align}
Y_j=\begin{pmatrix}
p_j(x) & u_j(x)\\
v_j(x)& q_j(x)\end{pmatrix},\quad j=1,..,n-2,
\end{align}
where $p_j,q_j,u_j,v_j$ are smooth functions of $x$. To guarantee \eqref{compitabile condition}, all the coefficients of $z$ must vanish. Those equations can be solved recursively. Eventually, by eliminating $u_j,v_j, j=2,..,n-2$, and let $v_1=u_1$, we will arrive at a nonlinear ODE of $u_1$\footnote{Surprisingly, the dependence on $p_j,q_j$ will disappear.}, which turns out to be a member of the hierarchy of Painlev\'e II equations . We list the first few of them:
\begin{align*}
n=3&: -8cu^3+cu_{xx}-4xu=0,\\
n=5&:-24cu^5+10cu^2u_{xx}+10cuu_x^2-\frac{c}{4}u_{xxxx}-4xu=0,\\
n=7&: -80cu^7+70cu^4u_{xx}+140cu^3u_x^2-\frac{7cu^2u_{xxxx}}{2}\\
\quad \quad & +(-\frac{21}{2}cu^2_{xx}-14cu_xu_{xxx}-4x)u+\frac{c}{16}u_{xxxxxx}-\frac{35}{2}cu_x^2u_{xx}=0.
\end{align*}
In the current article, we focus only on the odd members. In fact, $n=3$ corresponds to the mKdV equation, $n=5$ corresponds to the 5th order mKdV, and so on. In the following subsection, we will show how to connect the long-time asymptotics behavior of the mKdV hierarchy to the solutions to the Painlev\'e II hierarchy.

\subsection{Painlev\'e Region}
Recall the phase functions of the AKNS hierarchy of mKdV type equations are
\begin{align}\label{original phase function}
\theta(z;x,t)=xz+ctz^n, \quad n\, \text{is odd}. 
\end{align}
By the Painlev\'e region we mean a collection of all the curves $x=s(nt)^{1/n},s\neq 0$, by rescaling $z\rightarrow (nt)^{-\frac{1}{n}}\xi$, we have 
\begin{align}
\Theta(\xi)=s\xi+\frac{c}{n}\xi^n.
\end{align}
Now the modulus of the stationary phase points of \eqref{original phase function} is 
\begin{equation*}
|z_0|=\left|-\frac{x}{ct}\right|^{\frac{1}{n-1}}=\OO(t^{-\frac{1}{n}}),
\end{equation*} 
and however, after scaling, the modulus of the stationary phase points of $\Theta(\xi)$ is
\begin{align}
\xi_0=|z_0|t^{\frac{1}{n}},
\end{align}
which is fixed as $t\rightarrow \infty$.
A direct computation shows for any odd $n$, one can always perform lens-opening to the rays $\{z\in\R:|z|>|\xi_0|\}$, due to the signature of $\Re(i\theta)$, see Fig.\ref{signature of n=9}.
\begin{figure}[h]
\centering
\includegraphics[width=0.75\textwidth]{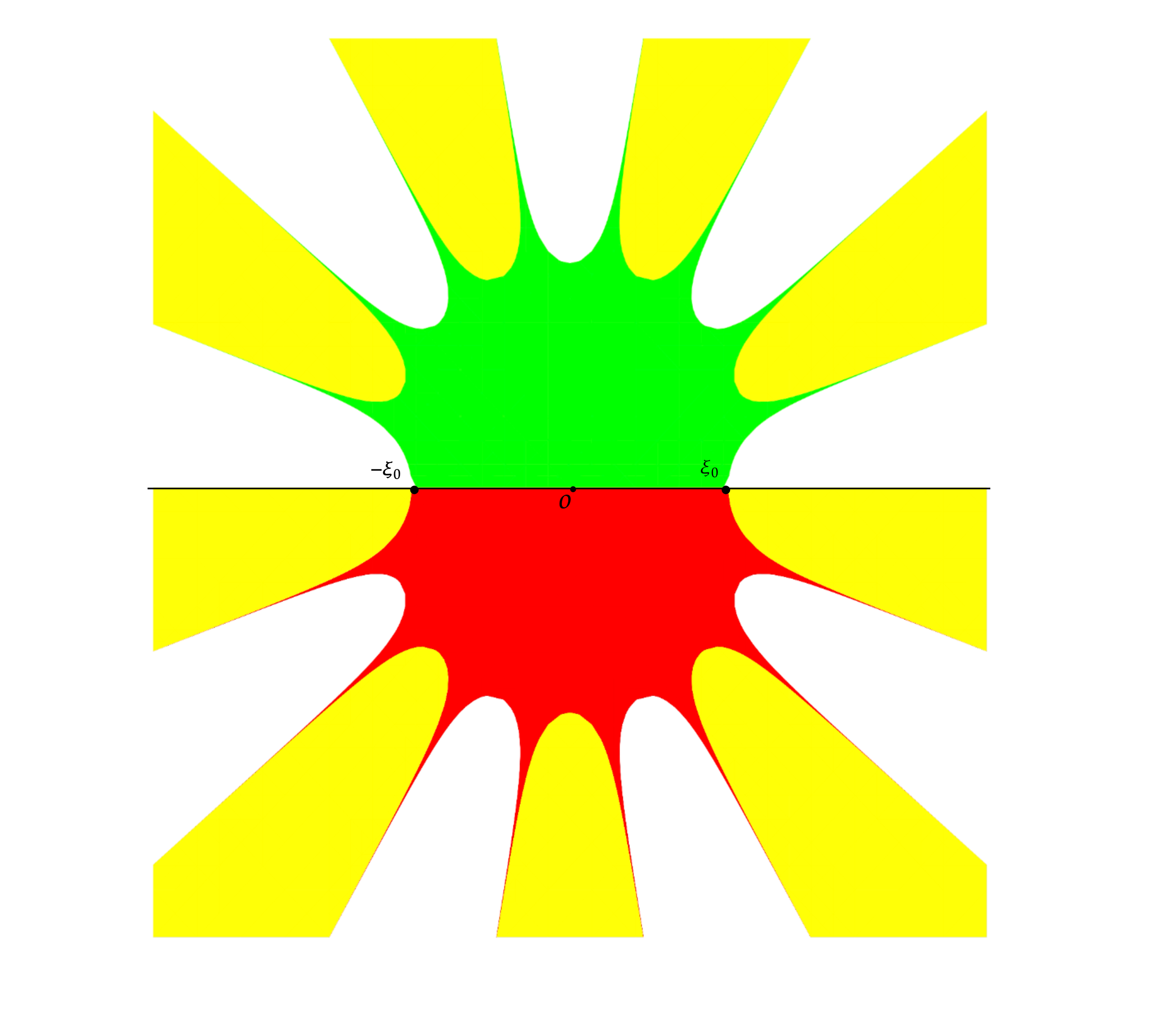}
\caption{Signature of $\Re(i\theta)$. The green region: $\Re(i\theta)>0$ when $x<0$; The red region: $\Re(i\theta)>0$ when $x>0$; The yellow region: the overlapping region of red and green; The white region: $\Re(i\theta)<0$. Here we only plot the signatures of $\Re{(i\theta)}$ for $n=9$. Other odd $n$, the region plot looks very similar.}
\label{signature of n=9}
\end{figure}

Note that
\begin{align*}
e^{-i\theta(z)\ad \sigma_3}v(z)&=e^{-i\Theta(\xi)\ad \sigma_3}v(\xi)\\
&=\begin{pmatrix}
1-|R|^2 & -\bar{R}e^{-2i\Theta}\\
Re^{2i\Theta} & 1
\end{pmatrix}\\
&=\begin{pmatrix}
1 & -\bar{R}e^{-2i\Theta}\\
0 & 1
\end{pmatrix}\begin{pmatrix}
1 & 0\\
Re^{2i\Theta} & 1
\end{pmatrix}.
\end{align*}
We can deform the the contour $\{z\in \R:|z|>|\xi_0|\}$ as before and get the deformed contour as follows (see Fig.\ref{painleve region contours}): Fix a positive constant $\alpha<\frac{\pi}{n}$\footnote{Such a choice of $\alpha$ guarantees that the new contours will stay within the regions where the corresponding exponential term will decay (considering Fig.\ref{signature of n=9}).},
\begin{align*}
\Sigma_0&=\{z\in \R:-\xi_0\leq z\leq\xi_0\},\\
    \Sigma_1&=\{z:z=\xi_0+\rho e^{i\alpha},\rho\in (0,\infty)\},\\
    \Sigma_2&=\{z:z=-\xi_0+\rho e^{-i\alpha},\rho\in (-\infty,0)\},\\
    \Sigma_3&=\{z:z=-\xi_0+\rho e^{i\alpha},\rho\in (-\infty,0)\},\\
    \Sigma_4&=\{z:z=\xi_0+\rho e^{-i\alpha},\rho\in (0,\infty)\},
\end{align*}
and we define the regions as follows:
\begin{align*}
    \Omega_1&=\{z:z=\xi_0+\rho e^{i\phi},\rho\in (0,\infty),\phi\in (0,\alpha)\},\\
    \Omega_2&=\C^+\backslash (\Omega_1\cup\Omega_3),\\
    \Omega_3&=\{z:z=-\xi_0+\rho e^{-i\phi},\rho\in (-\infty,0),\phi\in (-\alpha,0)\},\\
     \Omega_4&=\{z:z=-\xi_0+\rho e^{i\phi},\rho\in (-\infty,0),\phi\in (0,\alpha)\},\\
     \Omega_5&=\C^{-}\backslash (\Omega_4\cup\Omega_6),\\
    \Omega_6&=\{z:z=\xi_0+\rho e^{i\phi},\rho\in (0,\infty),\phi\in (-\alpha,0)\}.
\end{align*}

\begin{figure}[h]
\centering
\begin{tikzpicture}
\path [draw=black,postaction={on each segment={mid arrow=red}}] (-5,0)--(-2,0)--(2,0)--(5,0);
\node[circle, fill=red, scale=0.3]  at (2,0) {};
\node[circle, fill=red, scale=0.3]  at (-2,0) {};
\node[below left] at (2,0) {\small $\xi_0$};
\node[below right] at (-2,0) {\small $-\xi_0$};
 \path [draw=blue,postaction={on each segment={mid arrow=red}}] (2,0) -- (4.8,2); 
 \path [draw=blue,postaction={on each segment={mid arrow=red}}] (2,0) -- (4.8,-2); 
 \path [draw=blue,postaction={on each segment={mid arrow=red}}] (-4.8,2)--(-2,0) ; 
 \path [draw=blue,postaction={on each segment={mid arrow=red}}] (-4.8,-2)--(-2,0); 
 \node [right] at (4.2,1.5) {$\Sigma_1$};
 \node [left] at (-4.2,1.5) {$\Sigma_2$};
 \node [left] at (-4.2,-1.5) {$\Sigma_3$};
 \node [right] at (4.2,-1.5) {$\Sigma_4$};
 \node [below] at (0,0) {$\Sigma_0$};
 \node [right] at (3.2,0.5) {$\Omega_1$};
 \node [above] at (0,0.5) {$\Omega_2$};
 \node [left] at (-3.2,0.5) {$\Omega_3$};
 \node [left] at (-3.2,-0.5) {$\Omega_4$};
 \node [below] at (0,-1.2) {$\Omega_5$};
 \node [right] at (3.2,-0.5) {$\Omega_6$};
\end{tikzpicture}
\caption{Contour for $\dbar$-RHP.}
\label{painleve region contours}
\end{figure}
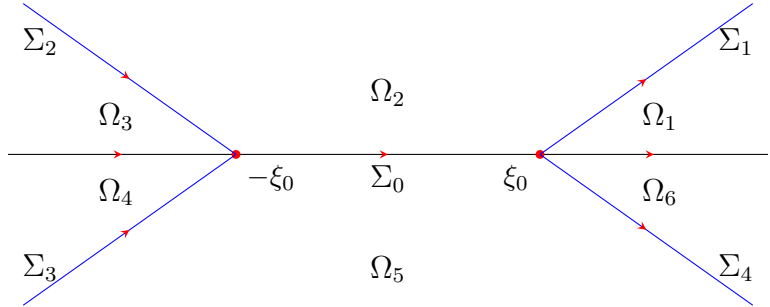
As before, set the original RHP as $m^{[1]}$ with jump $e^{-i\theta(z)\ad \sigma_3}v(z)$. After re-scaling and $\dbar$-lenses opening, we set $m^{[2]}(\xi)=m^{[1]}O(\gamma)$, where the lenses opening matrix is
\begin{align}
O(\gamma)=\begin{cases}
\begin{pmatrix}
1 & 0\\
-E_+e^{2i\Theta(\gamma)} & 1
\end{pmatrix},\quad \gamma\in \Omega_1\cup \Omega_3,\\
\begin{pmatrix}
1 & -E_-e^{-2i\Theta(\gamma)}\\
0 & 1
\end{pmatrix},\quad 
\gamma\in \Omega_4\cup \Omega_6,\\
I,\quad \gamma\in \Omega_2\cup \Omega_5,
\end{cases}
\end{align}
where
\begin{align*}
E_+(\gamma)&=\mathcal{K}(\phi)R\left((nt)^{-\frac{1}{n}}\xi\right)+(1-\mathcal{K}(\phi))R(\tilde{\xi}_0(nt)^{-\frac{1}{n}}),\\
E_-&(\gamma)=\overline{E_+(\gamma)},\\
\gamma&=
\begin{cases}
\xi_0+\rho e^{i\phi},\quad \text{if }\gamma\in \Omega_1\cup\Omega_6,\\
-\xi_0+\rho e^{i\phi},\quad \text{if }\gamma \in \Omega_3\cup\Omega_4,
\end{cases} \\
\xi&=\Re(\gamma),\\
\tilde{\xi}_0&=\begin{cases}
\xi_0,\quad \text{if }\gamma\in \Omega_1\cup\Omega_6,\\
-\xi_0,\quad \text{if }\gamma \in \Omega_3\cup\Omega_4.
\end{cases}
\end{align*}
Now we arrive at the following $\dbar$-RHP:
\begin{pproblem}
\begin{enumerate}[label={(\arabic*)}] Looking for a 2 by 2 matrix-valued function $m^{[2]}$ such that
\item{The RHP}:
\subitem(1.a) $m^{[2]}(\gamma)\in C^1(\R^2\backslash \Sigma)$ and $m^{[2]}(z)=I+\mathcal{O}(\gamma^{-1}),\gamma\rightarrow \infty$; 
\subitem(1.b) the jumps on $\Sigma_1$ and $\Sigma_2$ are $e^{-i\Theta(\xi)\ad \sigma_3}v_+$, and 
the jumps on $\Sigma_3$ and $\Sigma_4$ are $e^{-i\Theta(\xi)\ad \sigma_3}v_-$, where
\begin{align*}
    v_=\begin{pmatrix}
    1 & \bar{R}\\
    0 & 1
    \end{pmatrix},\quad v_+=\begin{pmatrix}
    1 & 0\\
    R & 1
    \end{pmatrix}.
\end{align*}
The jump on $\Sigma_0$ is
$e^{-i\Theta\ad\sigma_3}v((nt)^{-\frac{1}{n}}\xi),$, and the jumps on $\{z\in\R:|z|>|\xi_0|\}$ is $I$.

\item{The $\dbar$-problem}:

For $z\in \C$, we have
\begin{align}
\dbar m^{[2]}(\xi)=m^{[2]}(\xi)\dbar O(\xi).
\end{align}
\end{enumerate}
\end{pproblem}

Again, we will need the following lemma in order to estimate errors from the $\dbar$-problem.
\begin{lemma}
For $\gamma\in \Omega_{1,3,4,6}$, $\xi=\Re\gamma$,
\begin{align}
|\dbar E_{\pm}(\gamma)|\leq (nt)^{-\frac{1}{n}}|(nt)^{-\frac{1}{n}}(\xi-\xi_0)|^{-\frac{1}{2}}\|R\|_{H^{1,0}}+(nt)^{-\frac{1}{n}}|R'((nt)^{-\frac{1}{n}}\xi)|.
\end{align}
\end{lemma}
\begin{proof}
For brevity, we only prove for the region $\Omega_1$.
Using the polar coordinates, we have
\begin{align*}
|\dbar E_+(\gamma)|&=\left|\frac{ie^{i\phi}}{2\rho}\mathcal{K}'(\phi)\left[R\left((nt)^{-\frac{1}{n}}\xi\right)-R(\xi_0(nt)^{-\frac{1}{n}})\right]+\mathcal{K}(\phi)R'\left((nt)^{-\frac{1}{n}}\xi\right)(nt)^{-\frac{1}{n}}\right| \\
&\text{by Cauchy-Schwartz inequality}\\
&\leq \left|\frac{\|R\|_{H^{1,0}}|(nt)^{-\frac{1}{n}}\xi-\xi_0(nt)^{-\frac{1}{n}}|^{1/2}}{\gamma-\xi_0}\right|+(nt)^{-\frac{1}{n}}\left|R'\left((nt)^{-\frac{1}{n}}\xi\right)\right|\\
&\leq (nt)^{-\frac{1}{n}}|(nt)^{-\frac{1}{n}}(\xi-\xi_0)|^{-\frac{1}{2}}\|R\|_{H^{1,0}}+(nt)^{-\frac{1}{n}}|R'((nt)^{-\frac{1}{n}}\xi)|.
\end{align*}
Similarly, we can prove for other regions.
\end{proof}
Next, consider a pure RHP $m^{[3]}$ which satisfies exactly the RHP part of $\dbar$-RHP($m^{[2]}$). $m^{[3]}$ can be approximated by the RHP corresponding to a special solution of the Painlev\'e II hierarchy\footnote{As for the existence of the RHP $m^{[3]}$, which is not completely trivial due to the fact that solutions to the Painlev\'e II equations have poles,  we refer the readers to the book\cite{fokas2006painleve} for the details}. Since for $\gamma\in \Omega_1$,
\begin{align*}
&\left|\left(R(\xi (nt)^{-\frac{1}{n}})-R(0)\right)e^{2i\Theta(\gamma)}\right|\\
&\leq |\xi(nt)^{-\frac{2}{n}}|^{\frac{1}{2}}\|R\|_{H^{1,0}}e^{2\Re{i\Theta(\gamma)}}\\
&\leq (nt)^{-\frac{1}{n}}|\Re\gamma|^{\frac{1}{2}}\|R\|_{H^{1,0}}e^{2\Re{i\Theta(\gamma)}},
\end{align*}
it is evident that 
\begin{align}
\label{requirement for small norm}
\|Re^{2i\Theta}-R(0)e^{2i\Theta}\|_{L^{\infty}\cap L^1\cap L^2}\leq c(nt)^{-\frac{1}{n}}.
\end{align}
Let $m^{[4]}$ solves the RHP formed by replacing $R(\pm\xi_0(nt)^{-1/n})$ and its complex conjugate in the jumps of $m^{[3]}$ along $\Sigma_k,k=1,2,3,4$ by $R(0)$ and $\bar{R}(0)$ respectively. Then, by the small norm theory, the errors between the corresponding potential is given by
\begin{align*}
error_{3,4}&=\lim_{\gamma\rightarrow\infty}|\gamma(m^{[4]}_{12}-m^{[3]}_{12})|\\
&\leq c \int_{\Sigma} |(R(\Re(s) (nt)^{-\frac{1}{n}})-R(0))e^{2i\Theta(s)}|ds\\
&\leq c(nt)^{-\frac{1}{n}}.
\end{align*}
Then since now the jumps are all analytic, we can perform an analytic deformation and arrive at the green contours as show in Fig.\ref{Countors for Painleve II hierarchy}. Let's denote the new RHP by $m^{[5]}(\gamma)$, and we arrive at the following RHP:
\begin{problem}Looking for a 2 by 2 matrix-valued function $m^{[5]}$ such that 
\begin{enumerate}[label={(\arabic*)}]
    \item $m^{[5]}$ is analytic off the contours $\cup_{k=1,2,3,4}\Sigma_{k}^{[5]}$;
    \item $m^{[5]}_+=m^{[5]}_-v^{[5]},\quad z\in \cup_{k=1,2,3,4}\Sigma_{k}^{[5]},$ where
    \begin{align*}
        v^{[5]}=\begin{cases}
        \begin{pmatrix}
1 & 0\\
R(0)e^{2i\Theta(\gamma)}&1
\end{pmatrix},\quad \gamma\in \Sigma_{1}^{[5]}\cup \Sigma_2^{[5]},\\
\begin{pmatrix}
1 &\bar{R}(0)e^{-2i\Theta(\gamma)}\\
0&1
\end{pmatrix},\quad \gamma\in \Sigma_{3}^{[5]}\cup \Sigma_4^{[5]}.
        \end{cases}
    \end{align*}
\end{enumerate}
Here the new contours (see Fig.\ref{Countors for Painleve II hierarchy}) are
\begin{align*}
    \Sigma_1^{[5]}&=\{z:z=\rho e^{i\alpha},\rho\in (0,\infty)\},\\
    \Sigma_2^{[5]}&=\{z:z=\rho e^{-i\alpha},\rho\in (-\infty,0)\},\\
    \Sigma_3^{[5]}&=\{z:z=\rho e^{i\alpha},\rho\in (-\infty,0)\},\\
    \Sigma_4^{[5]}&=\{z:z=\rho e^{-i\alpha},\rho\in (0,\infty)\}.
\end{align*}
\end{problem}

\begin{figure}[h]
\centering
\begin{tikzpicture}
\path [draw=black,postaction={on each segment={mid arrow=red}}] (-2,0)--(2,0);
\node[circle, fill=red, scale=0.3]  at (2,0) {};
\node[circle, fill=red, scale=0.3]  at (-2,0) {};
\node[below left] at (2.2,0) {\small $|\xi_0|$};
\node[below right] at (-2.2,0) {\small $-|\xi_0|$};
 \path [draw=blue,postaction={on each segment={mid arrow=red}}] (2,0) -- (4.8,2); 
 \path [draw=green,postaction={on each segment={mid arrow=red}}] (0,0) -- (2.8,2); 
 \path [draw=blue,postaction={on each segment={mid arrow=red}}] (2,0) -- (4.8,-2); 
 \path [draw=green,postaction={on each segment={mid arrow=red}}] (0,0) -- (2.8,-2); 
 \path [draw=blue,postaction={on each segment={mid arrow=red}}] (-4.8,2)--(-2,0) ; 
 \path [draw=blue,postaction={on each segment={mid arrow=red}}] (-4.8,-2)--(-2,0); 
 \path [draw=green,postaction={on each segment={mid arrow=red}}] (-2.8,2)--(0,0) ; 
 \path [draw=green,postaction={on each segment={mid arrow=red}}] (-2.8,-2)--(0,0); 
 \node [right] at (4.2,1.5) {$\Sigma_1$};
 \node [right] at (2.2,1.5) {$\Sigma_1^{[5]}$};
 \node [left] at (-4.2,1.5) {$\Sigma_2$};
 \node [left] at (-2.2,1.5) {$\Sigma_2^{[5]}$};
 \node [left] at (-4.2,-1.5) {$\Sigma_3$};
  \node [left] at (-2.2,-1.5) {$\Sigma_3^{[5]}$};
 \node [right] at (4.2,-1.5) {$\Sigma_4$};
  \node [right] at (2.2,-1.5) {$\Sigma_4^{[5]}$};
\end{tikzpicture}
\caption{Contour for $m^{[4]}$(Green part).}
\label{Countors for Painleve II hierarchy}
\end{figure}
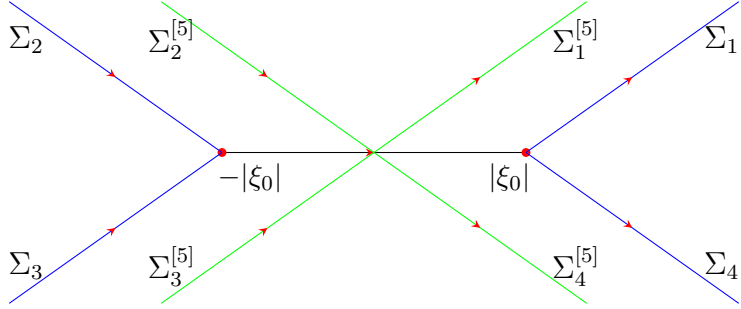

Then according to the previous subsection, the $(1,2)$ entry of the solution $m^{[5]}$, similarly the solution $m^{[4]}$, is the solution to the Painlev\'e II hierarchy, i.e.,
\begin{align}
m^{[4]}_{12}(\gamma)=m^{[5]}_{12}(\gamma),
\end{align}
Hence we have $P^{II}_k(s)=\lim_{\gamma\rightarrow\infty}\gamma m^{[5]}_{12}$
where $P^{II}_k$ solves the $k^{th}$ equation in the Painlev\'e II hierarchy, where $k=\frac{n-1}{2}$.

Now let's consider the error generated from the $\dbar$-extension. Recall that the error $E$ satisfies a pure $\dbar$-problem:
\begin{align*}
\dbar E&=EW,\\
W&=m^{[3]}\dbar O(m^{[3]})^{-1}.
\end{align*}
As before, the $\dbar$-equation is equivalent to an integral equation which reads
\begin{align*}
E(z)=I+\frac{1}{\pi}\int_{\C}\frac{E(s)W(s)}{z-s}\text{d}A(s)=I+\mathcal{K}(E).
\end{align*}
As before, we can show that the resolvent always exists for large $t$. So we only need to estimate the true error which is: $\lim_{z\rightarrow \infty}z(E-I)$. In fact, we have
\begin{align*}
\lim_{z\rightarrow\infty}|z(E-I)|&=|\int_\mathbb{C}EWds|\\
&\leq c\|E\|_\infty\int_{\Omega}|\dbar O|ds.
\end{align*}
For the sake of simplicity, we only estimate the integral on the right hand side in the region of the top right corner. Note there is only one entry which is nonzero in $\dbar O$, which is one of the $E_{\pm}$ and we split the integral into two parts in the obvious way, i.e.,
\begin{align*}
\int_\Omega |\dbar O|ds&\leq I_1+I_2\\
&=\int_\Omega (nt)^{-\frac{1}{2n}}|\Re{s}-\xi_0|\|R\|_{H^{1,0}}e^{2\Re{i\Theta(s)}}ds\\
&+\int_{\Omega}(nt)^{-\frac{1}{nt}}|R'((nt)^{-\frac{1}{n}}s)|e^{2\Re{i\Theta(s)}}ds.
\end{align*}
As we know from previous sections, $e^{\Re{2i\Theta(s)}}\leq ce^{-2|\Theta''(\xi_0)|uv}$ in the region $\{z=u+iv:u>\xi_0,0<v<\alpha u\}$ for some small $\alpha$ , where $s=u+iv+\xi_0$. Then we have
\begin{align*}
I_1&\leq (nt)^{-\frac{1}{2n}}\int_{\Omega}|\Re{s}-\xi_0|^{-1/2}e^{-cuv}dudv\\
&\leq (nt)^{-\frac{1}{2n}}\int_0^\infty\int_0^{\alpha u}u^{-1/2}e^{-cuv}dudv\\
&\leq C(nt)^{-\frac{1}{2n}}\int_0^\infty\frac{1-e^{-2\alpha|\Theta''(\xi_0)|} }{u^{3/2}}du\\
&=\OO\left((nt)^{-\frac{1}{2n}}\right),
\end{align*}
and 
\begin{align*}
I_2&\leq (nt)^{-\frac{1}{n}}\int |R'((nt)^{-\frac{1}{2n}}\Re{s})|e^{-cuv}dudv\\
&\text{ by Cauchy-Schwartz inequality}\\
&\leq (nt)^{-\frac{1}{n}}\|R\|_{H^{1,0}}\int_0^\infty(\int_{\alpha v}^\infty e^{-2cuv}du)^{1/2}dv\\
&\leq (nt)^{-\frac{1}{n}}\|R\|_{H^{1,0}}\int_0^\infty \frac{e^{-c\alpha v^2}}{\sqrt{2\alpha cv}}dc\\
&=\OO((nt)^{-\frac{1}{n}}).
\end{align*}
Thus, we arrive at 
\begin{align}
\dbar\text{Error} = \OO((nt)^{-\frac{1}{2n}}). 
\end{align}
And we undo all the deformations, we obtain
\begin{align*}
m^{[1]}((nt)^{-\frac{1}{n}}\gamma)&=m^{[2]}(\gamma)O^{-1}(\gamma)\\
&=(1+\frac{\OO{(t^{\frac{1}{2n}})}}{\gamma})m^{[3]}(\gamma)O^{-1}(\gamma)\\
&=(1+\frac{\OO{(t^{\frac{1}{2n}})}}{\gamma})(1+\frac{\OO{(t^{\frac{1}{2n}})}}{\gamma})m^{[4]}(\gamma)O^{-1}(\gamma)
\\
&=(1+\frac{\OO{(t^{\frac{1}{2n}})}}{\gamma})(1+\frac{\OO{(t^{\frac{1}{2n}})}}{\gamma})m^{[5]}(\gamma)O^{-1}(\gamma).
\end{align*}
It can also be rewritten in terms of the variable $z$:
\begin{align*}
m^{[1]}(z)=\left(1+\frac{\OO{(t^{-1/(2n)}})}{z(nt)^{1/n}}\right)m^{[5]}((nt)^{1/n}z)+\OO{(z^{-2})},\quad z\rightarrow\infty.
\end{align*}
Since $m^{[5]}$ corresponds to the RHP for the Painlev\'e II hierarchy, we have 
\begin{align*}
m^{[5]}(\gamma)=I+\frac{m_1^{[5]}(s)}{\gamma}+\OO(\gamma^{-1}),
\end{align*}
where $\gamma=z(nt)^{1/n}$.

Thus,
\begin{align*}
m^{[1]}(z)&=\left(1+\frac{\OO{(t^{-\frac{1}{2n}})}}{z(nt)^{1/n}}\right)\left(1+\frac{m_1^{[5]}(s)}{z(nt)^{1/n}}+\OO(z^{-2})\right)\\
&=I+\frac{m_1^{[5]}(s)}{z(nt)^{1/n}}+\frac{\OO{(t^{-\frac{1}{2n}})}}{z(nt)^{1/n}}+\OO(z^{-2}).
\end{align*}
Since $m_1^{[5]}(s)$ is connected to solutions of the Painlev\'e II hierarchy, we conclude that
\begin{align*}
q(x,t)&=\lim_{z\rightarrow\infty}z(m^{[1]}-I)\\
&=(nt)^{-\frac{1}{n}}u_n(x(nt)^{-\frac{1}{n}})+\OO{(t^{-\frac{3}{2n}})},
\end{align*}
where $u_n$ solves the $\frac{n-1}{2}$th equation of the Painlev\'e II hierarchy. The odd integer $n$ corresponds to the $\frac{n-1}{2}$th member in the mKdV hierarchy.
\begin{remark}
As for the asymptotics for the Painlev\'e II equation, we refer the readers to the classical book \cite{fokas2006painleve}. There are also some recent works related to Painlev\'e II hierarchy, see for example \cite{miller_rational_2017},\cite{tom2010},\cite{cafasso2021}.
\end{remark}
\bibliographystyle{abbrv}
\bibliography{References}


\end{document}